\documentclass[runningheads]{llncs}

\usepackage[T1]{fontenc}

\usepackage{amssymb}
\usepackage{amsmath}

\usepackage{amsfonts}

\usepackage{amsthm}
\usepackage{mathtools}

\usepackage{stmaryrd}
\usepackage{xcolor, soul}

\usepackage{graphicx}
\usepackage{bbding}
\usepackage{pifont}
\usepackage{paralist}
\usepackage{enumitem}
\usepackage{lineno}
\usepackage{wrapfig}
\usepackage{xspace}
\usepackage[shortcuts]{extdash}
\usepackage{anyfontsize}
\usepackage{multicol}
\usepackage{blindtext}

\usepackage[yyyymmdd,hhmmss]{datetime}
\usepackage{hyperref}
\usepackage{cleveref} 
\usepackage{thmtools} 
\usepackage{thm-restate}
\DeclareMathAlphabet{\mymathbb}{U}{bbold}{m}{n}

\usepackage{tikz}
\usetikzlibrary{arrows,positioning,shapes,decorations,automata,backgrounds,petri,fit,calc,shapes.multipart,decorations.text,calc,arrows.meta}
\usepackage{relsize}

\newcounter{Reqs}
\Roman{Reqs}

\counterwithin{figure}{section}

\counterwithin{theorem}{section}
\counterwithin{lemma}{section}
\counterwithin{corollary}{section}
\counterwithin{proposition}{section}
\counterwithin{exercise}{section}
\counterwithin{definition}{section}
\counterwithin{conjecture}{section}
\counterwithin{example}{section}
\counterwithin{remark}{section}
\counterwithin{note}{section}
\counterwithin{figure}{section}

\AtBeginDocument{%

}

\newcommand{\fdpa}{\textsc{d}DPA}

\tikzset{fontscale/.style = {font=\relsize{#1}}}

\newcommand{\commentout}[1]{}

\newcommand{\aggregator}[1]{\ensuremath{\textsf{#1}\xspace}}
\renewcommand{\sup}{\aggregator{Sup}}
\renewcommand{\inf}{\aggregator{Inf}}
\renewcommand{\limsup}{\aggregator{LimSup}}
\renewcommand{\liminf}{\aggregator{LimInf}}
\newcommand{\limsupavg}{\aggregator{LimSupAvg}}
\newcommand{\liminfavg}{\aggregator{LimInfAvg}}
\newcommand{\limavg}{\aggregator{LimAvg}}

\newcommand{\val}{\aggregator{Val}}
\newcommand{\valrbst}{\aggregator{Rbst}}

\newcommand{\true}{\textsc{t}}
\newcommand{\false}{\textsc{f}}
\newcommand{\inft}{\textsl{inft}}
\newcommand{\noval}{\textrm{-}}

\newcommand{\ltlG}{\ensuremath{\textbf{G}}}

\newcommand{\score}{\ensuremath{\textsl{score}}}

\newcommand{\avgscore}{\ensuremath{\textsl{avgscore}}}

\newcommand{\lcolor}{\ensuremath{\textsl{color}}}

\newcommand{\icolors}{\ensuremath{\textsl{colors}}}

\newcommand{\any}{any}

\newcommand{\arop}{\hookrightarrow}

\newcommand{\gtrbst}{\mathrel{\rhd}}

\newcommand{\sema}[1]{\ensuremath{\llbracket#1\rrbracket}}
\newcommand{\aut}[1]{\mathcal{#1}}

\newcommand{\truerobustness}{\text{\tiny{\FourStar}}}

\newcommand{\term}[1]{\textcolor{brown}{\textsl{#1}}}
\newcommand{\mathterm}[1]{\textcolor{brown}{{#1}}}

\newcommand{\col}[1]{\ensuremath{\textsl{#1}}}

\newcommand{\alphabetfont}[1]{\textsc{#1}}
\newcommand{\green}{\alphabetfont{g}}	
\newcommand{\red}{\alphabetfont{r}}		
\newcommand{\yellow}{\alphabetfont{y}}
\newcommand{\black}{\alphabetfont{b}}
\newcommand{\white}{\alphabetfont{w}}

\newcommand{\Sigmasys}{\widehat\Sigma}
\newcommand{\Sigmaprop}{\Sigma}

\newcommand{\lassovalue}{\ensuremath{\mu_{\aut{A}}^{\circlearrowright}}}
\newcommand{\prefixvalue}{\ensuremath{\mu_{\aut{A}}^{\leadsto}}}
\newcommand{\bestrecom}{\ensuremath{\Sigma_{\aut{A}}^\leadsto}}
\newcommand{\finerecom}{\ensuremath{\Sigma_{\aut{A}}^\circlearrowright}}
\newcommand{\alg}[1]{\ensuremath{\mathbf{#1}}}
\newcommand{\simple}{\textsl{simple}}
\newcommand{\trail}{\textsl{trail}}

\newcommand{\signal}[1]{\texttt{#1}}
\newcommand{\sigqone}{\signal{QUE1}}
\newcommand{\sigqtwo}{\signal{QUE2}}

\tikzset{pyellow/.style={preaction={ 
draw,yellow,-, 
double=yellow,
line width=1.1pt,
}}}
\tikzset{pgreen/.style={preaction={ 
draw,green!60,-, 
line width=1.1pt,
double=green!60,
}}}
\tikzset{pred/.style={preaction={ 
draw,magenta!60,-, 
double=magenta!60,
line width=1.1pt,
}}}

\tikzset{pblack/.style={preaction={ 
draw,gray!80,-, 
double=gray!80,
line width=1.1pt,
}}}

\tikzset{pwhite/.style={preaction={ 
draw,gray!20,-, 
double=gray!20,
line width=1.1pt,
}}}

\tikzset{pblue/.style={preaction={ 
draw,blue!25,-, 
double=blue!25,
line width=1.3pt,
}}}

\tikzset{porange/.style={preaction={ 
draw,orange!50,-, 
double=orange!50,
line width=1.3pt,
}}}

\tikzset{pbrown/.style={preaction={ 
draw,brown!50,-, 
double=brown!50,
line width=1.3pt,
}}}

\tikzset{pteal/.style={preaction={ 
draw,teal!50,-, 
double=teal!50,
line width=1.3pt,
}}}

\tikzset{
    side by side/.style 2 args={
    line width=5pt, -,
    #2, 
    postaction={
        clip,postaction={draw,#1}
        }
    }
}

\tikzset{
ministate/.style={align=center,
        state,inner sep=2pt, minimum size=3pt},
}

\tikzset{
smallstate/.style={align=center,
        state,inner sep=5pt, minimum size=5pt},
}

\tikzset{
nostate/.style={align=center,
        label,inner sep=2pt, minimum size=3pt},
}

\begin{document}

\title{Runtime Consultants}

\titlerunning{Runtime Consultants}

\newcommand{\orcidicon}[1]{\href{https://orcid.org/#1}{\includegraphics[width=10pt]{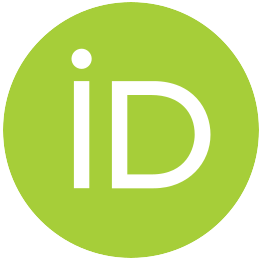}}}
\author{Dana Fisman\orcidicon{0000-0002-6015-4170} \and
Elina Sudit\thanks{ 
Supported by ISF grant 2507/21 and Frankel Center for Computer Science, BGU.}\orcidicon{0009-0009-6187-6894}}

\authorrunning{D. Fisman and E. Sudit}
%
\institute{Ben-Gurion University\\
\email{dana@bgu.ac.il}  \quad
\email{elinasu@post.bgu.ac.il}}
\maketitle              

\begin{abstract}
In this paper we introduce the notion of a \emph{runtime consultant}. A runtime consultant is defined with respect to some value function on infinite words. Similar to a runtime monitor, it runs in parallel to an execution of the system and provides inputs at every step of the run. While a runtime monitor alerts when a violation occurs, the idea behind a consultant is to be \emph{pro-active} and provide recommendations for which action to take next in order to avoid violation (or obtain a maximal value for quantitative objectives). It is assumed that a \emph{runtime-controller} can take these recommendations into consideration.
The runtime consultant does not assume that its recommendations are always followed. Instead, it adjusts to the actions actually taken (similar to a vehicle navigation system). We show how to compute a runtime consultant for common value functions used in verification, and that almost all have a runtime consultant that works in constant time. 
We also develop consultants for $\omega$-regular properties, under both their classical Boolean semantics and their recently proposed quantitative interpretation.
\end{abstract}

\section{Introduction}\label{sec:intro}
In \emph{runtime verification}~\cite{lee1999runtime,BarringerGHS04,FalconeHR13,Maler16,BartocciDDFMNS18,HavelundP23} 
the system under test is monitored   against its formal specification, so that if and when it can be deduced that the current execution of the system violates the specification, an alert can be issued and measures to prevent further escalation can be taken. In many cases, feedback that the specification has been violated might  be too late in the sense that preventing further escalation is meaningless as the bad thing has already happened. For this reason, a recent thread of research in runtime verification concerns predictive analysis, where one tries to provide indication for possible failures as early as possible~\cite{ZhangMPL11,OmerP23,ZhaoHFDL24}. 
Taking this a step forward, in this work we introduce the notion of a \emph{runtime consultant}.

The idea of a \emph{runtime consultant} is to generalize the notion of a runtime monitor so that instead of alerting when a problem occurred, it provides a recommendation of which actions are preferable to take next in order to satisfy the property (and so avoid violation before it occurs). In this setting it is assumed that a \emph{runtime controller} has some control over the next action and can take the consultant's recommendation into account. The consultant, like a GPS navigation system, does not assume its recommendations are always followed. Instead, it adjusts along the way and simply recommends the best actions from the current point of the run.

One motivation for a runtime consultant (RC) is to have an entity that can assist in systems involving human decision-making by offering proactive, real-time recommendations toward achieving qualitative or quantitative goals. In clinical settings, for instance during surgery or emergency care, an RC might suggest immediate actions such as adjusting drug dosage, initiating a procedure, or switching 
protocols in response to changes in vital signs—helping ensure guideline adherence, patient stability, and risk minimization. In pilot-assistance 
systems or semi-autonomous vehicles, the RC can recommend maneuvers like decelerating, changing lanes, or rerouting, guided by qualitative goals (e.g., avoiding collisions, complying with airspace rules) and quantitative ones (e.g., minimizing fuel use or turbulence). 
In all cases, the RC continuously adapts to the operator’s actual decisions, much like a navigation system recalibrates after a missed turn. This paper lays the theoretical foundations for a runtime consultant, considering common value functions used in formal verification.

The definition of a runtime consultant assumes some value function by which some runs are determined better than others. The value function can be qualitative (Boolean) or quantitative (with many and even infinitely many values). 
Based on the value function and the current state of the execution, the runtime consultant gives recommendation on which action to take next in order to achieve the best value in the long run (given the history of the execution so far).

In this work we formally define the notion of a runtime consultant  for a given value function. We assume higher values are preferable. Two types of runtime consultants are proposed: the $\leadsto$-RC that gives recommendations for obtaining the highest achievable value in the overall system and the $\circlearrowright$-RC that gives recommendations for obtaining the highest value achievable on an execution that loops back to  the current state. We say that a runtime consultant is \emph{strong} if it makes no arbitrary choices in its recommendation; that is, it recommends \emph{all} actions from the current state that can ultimately result in the highest achievable value. Otherwise we term it \emph{weak}.

We prove some general properties on arbitrary RCs as well as Boolean RCs (i.e. RCs for a Boolean value function). 
We examine the widely used limit operators \sup, \inf, \limsup, \liminf\ and \limavg\ 
as value functions; and show how to calculate  runtime consultants for them. 
With regard to $\circlearrowright$-RC we show that for $\sup$ and $\limsup$ it is impossible to construct a strong $\circlearrowright$-RC in the general case, but for all other limit operators it is possible. With regard to $\leadsto$-RC only for $\inf$ it is possible to construct a strong RC. 
The construction of these RCs can be done in polynomial time for all operators except for $\limavg$ in the case of $\circlearrowright$-RC, for which we show the problem is coNP-complete.

The construction of the RC is done at the preprocessing step. The interesting parameter is the runtime of the constructed RC. To keep up with the monitored run, it is desirable that the running time would be constant. 
Regarding the runtime of the RCs for the above operators, the picture is as follows. They are constant in all cases but the ones for $\sup$ and $\inf$ in which they are logarithmic in the size of the maximal weight in the graph, which is also quite reasonable. 

We then turn to discuss RCs for $\omega$-regular properties, given by some parity automaton. 
We show that in the qualitative setting strong RCs may not exist in the general case; the {preprocessing} time is polynomial; and the running time is constant as desired. 

However, strong runtime consultants can be obtained for $\omega$-regular properties if we adopt a quantitative rather than qualitative value function. The quantitative interpretation of $\omega$-regular properties introduced in~\cite{FismanS25} defines a value function, $\valrbst_L$, which assigns each lasso word $w$ a robustness score $\valrbst_L(w)$, reflecting how robust $w$ is with respect to the property $L$. The range of $\valrbst_L(w)$ is infinite, distinguishing accepted words from rejected ones, while also capturing nuances based on both the periodic and transient parts of the word (with greater weight on the periodic part). We show that strong RCs can be constructed for this quantitative robustness function, with a coNP-complete preprocessing step. Once constructed, the RC operates in constant time, as required.

\subsubsection*{Related Work}
The concept of a runtime consultant  bears some resemblance to runtime enforcement mechanisms such as \emph{enforcement monitors} (EMs)~\cite{Schneider00,FalconeMFR11,PinisettyPTJFM16} \cite{LanotteMM20,HubletLBKT24}, but introduces several key differences. 
An EM is designed to handle untrusted programs by actively modifying their outputs—through suppression, termination, or insertion—to ensure compliance with a specification. In contrast, an RC acts passively: it does not interfere with the system’s actions but issues recommendations for the next step, which the system may or may not follow.

Technically, an EM observes a sequence of outputs $\sigma_1,...,\sigma_i$ produced by the system and eventually emits a modified sequence $\sigma'_1,...,\sigma'_j$ that satisfies the specification 
while minimally deviating from the original sequence. 
An RC, by contrast, on observing the same prefix $\sigma_1,...,\sigma_i$ outputs a set of recommended next actions $\Sigma_{next}$, from which it is recommended to choose the next action $\sigma_{i+1}$ in order to achieve correct or optimal behavior.

While EMs were originally conceived for enforcing security properties in sequential programs, they are less suitable for reactive systems, where interventions like buffering, suppression, or delaying actions are often infeasible. Shield synthesis was introduced to address this gap by adapting enforcement to reactive settings, enforcing safety properties without delays or skipped steps~\cite{BloemKKW15}.

However, both EMs and shield synthesis are fundamentally limited to safety properties~\cite{FalconeFM08,LigattiR10,DolzhenkoLR15}. While predictive EM can handle some liveness properties~\cite{PinisettyPTJFM16}, it requires partial knowledge of the system’s language—effectively treating the system as a gray-box. 
In contrast, runtime consultants are applicable to both safety and liveness properties, and operate in a fully black-box setting: they make no assumptions about the system’s structure or behavior. Moreover, RCs naturally generalize beyond Boolean specifications, supporting quantitative value functions—a capability that, to our knowledge, has not been addressed by prior runtime enforcement approaches.

For space restrictions some proofs are deferred to the appendix. 

\section{Preliminaries}\label{sec:prelim}
The set of all infinite words over $\Sigma$ is denoted $\Sigma^\omega$. A word $w{\in}\Sigma^\omega$ of the form $w{=}u(v)^\omega$ for some $u{\in}\Sigma^*$ and $v{\in}\Sigma^+$ is termed a \term{lasso word}.
We say that $u$ is the \term{spoke} of the lasso word and $v$ is its \term{period}. 
For the lasso word $u(v)^\omega$ we denote the pair $(u,v)$ for $u{\in}\Sigma^*$ and $v{\in}\Sigma^+$ as the \term{representation} of the lasso word. 

$\mathbb{B}$ abbreviates $\{\true,\false\}$. 
For $i\leq j$, we write $[i..j]$ for the set $\{i,i{+}1,...,j\}$.
For a word $v=\sigma_1\sigma_2...\sigma_m$, the infix of $v$ that starts at $\sigma_i$ and ends at $\sigma_j$ inclusive is termed $v[i..j]$. Similarly, the prefix (resp. suffix) of $v$ that ends (resp. starts) in $\sigma_i$ inclusive is termed $v[..i]$ (resp. $v[i..]$). 
The notation $x\preceq y$ (resp. $x\prec y$) denotes that $x$ is a prefix (resp. proper prefix) of $y$.

Since we are interested in both quantitative and qualitative (Boolean) value functions we introduce a computational model that generalizes both.
A \term{quantitative automaton} (abbr. \term{quatomaton}) $\aut{A}$ is a tuple $(\Sigma,Q,q_0,\delta,\kappa,\alpha)$ where $\Sigma$ is an alphabet, $Q$ a finite set of states, $q_0\in Q$ the initial state, 
and $\alpha:(\mathbb{T}_{\delta}\times\mathbb{T}_{\kappa})^\omega\to\mathbb{T}$ is the objective. 
A run of $\aut{A}$ on a word $w=\sigma_1\sigma_2\ldots\in\Sigma^\omega$ is 
a sequence  $\rho=q_0\xrightarrow[t_1,d_1]{\sigma_1}q_1\xrightarrow[t_2,d_2]{\sigma_2}q_2\cdots$ 
such that $\delta(q_{i-1},\sigma_i)=(q_i,t_i)$ and $\kappa(q_{i})=d_i$. The value of the run $\rho$ is $\alpha((t_1,d_1)(t_2,d_2)\cdots)$.
The value that $\aut{A}$ gives $w$, denoted $\mathterm{\sema{\aut{A}}(w)}$, is the value of the unique run of $\aut{A}$ on $w$. 
We say that $\aut{A}$ \mathterm{implements} $\val:\Sigma^\omega\to\mathbb{T}$ if for every $w\in\Sigma^\omega$ we have $\sema{\aut{A}}(w)=\val(w)$.
For a finite word $u\in\Sigma^*$ we use $\aut{A}(u)$ to denote the state that $\aut{A}$ arrives at on reading $u$.

A \term{labeled weighted graph} is a quantitative automaton where $\mathbb{T}_{\kappa}$ is a singleton and thus plays no role. 
A \term{parity automaton} is a quantitative automaton where $\mathbb{T}_{\kappa}$ is a finite set of integers, $\mathbb{T}_{\delta}$ is a singleton (hence plays no role), $\mathbb{T}=\mathbb{B}$
and $\alpha((\noval,d_1),(\noval,d_2),\ldots)$ is $\true$ iff $\min\inft(d_1,d_2,\ldots)$ is even,
where $\mathterm{\inft(d_1,d_2,\ldots)}$ is the set of items occurring infinitely often in the sequence, namely $ \{ d \in \mathbb{T}_{\kappa} ~|~ \forall i\!\in\!\mathbb{N}.\ \exists j\!>\!i.\ d_j=d\}$.
We use DPA for deterministic parity automata. Note that DPA can express all $\omega$-regular languages~\cite{PerrinPinBook}.

We associate with $\aut{A}$ its automaton graph $G{=}(V,E)$ where $V{=}Q$ and 
$E=\{(v,v')~|~\allowbreak\delta(v,\sigma){=}(v',t)$ for some $\sigma {\in} \Sigma, t{\in}\mathbb{T}_\delta\}$. 
A subset $C$ of vertices (i.e., states) is a \term{strongly connected component} (SCC) if there exists a non-empty path between any two vertices in $C$. 
An SCC is a \term{maximal strongly connected component} (MSCC) if there is no SCC $C'$ such that $C{\subsetneq} C'$. Every run of an automaton $\aut{A}$ on a word $w$ eventually stays within a single MSCC (i.e. visits no state outside of this MSCC). We term this MSCC the \term{final} MSCC of $w$ wrt $L$. A (possibly non-simple) path (resp. cycle) $\rho{=}e_1e_2{\ldots} e_m$ for edges $e_i{\in} E$ is termed a \term{trail} (resp. \term{cyclic trail}) if no edge occurs more than once (i.e., $e_i{\neq} e_j$ for $i,j{\in}[1..m]$, $i{\neq} j$).  

The following value functions on infinite words are defined wrt a weighted labeled graph. 
The \term{supremum} (resp. \term{infimum}), abbreviated $\sup$ (resp. $\inf$), value of an infinite word $w$ is
$\sup: \Sigma^\omega \to \mathbb{R}$  for $\mathrm{sup}_{i\in\mathbb{N}}\{\theta(w,i)\}$ (resp. $\inf: \Sigma^\omega \to \mathbb{R}$  for $\mathrm{inf}_{i\in\mathbb{N}}\{\theta(w,i)\}$) where $\theta(w,i)$ is the weight of the corresponding graph edge that the run visits by reading the $i$-th letter of $w$.
The \term{limit supremum} (resp. \term{limit infimum}), abbreviated $\limsup$ (resp. $\liminf$), value of an infinite word $w$ 
is defined as $ \lim_{n \to\infty} \mathrm{sup}_{i\geq n}\{\theta(w,i)\}$ (resp. $\lim_{n \to\infty} \mathrm{inf}_{i\geq n}\{\theta(w,i)\}$). 
Limit average value functions on infinite words, $\liminfavg$ and $\limsupavg$ for short, are defined as $\limsup_{n\to\infty}\frac{1}{n}\sum_{i=1}^{n} \theta(w,i)$ and 
$\liminf_{n\to\infty}\frac{1}{n}\sum_{i=1}^{n} \theta(w,i)$. Since we consider lasso words and on lasso words  $\liminfavg$ and $\limsupavg$ coincide, we do not distinguish between them and simply use $\limavg$.

\section{Definition of a consultant}\label{sec:def-consultant}
As mentioned in the introduction, a runtime consultant can be thought of as extending the idea of a runtime monitor --- instead of alerting when something bad has happened which might be too late, a runtime consultant should be proactive and  recommend what action to take next, in order to avoid violation, and in the long run satisfy the property.
Similar to a runtime monitor for a given property $L$, a runtime consultant bases its decision taking into account the history of the run so far. Clearly, a runtime consultant should recommend taking actions that avoid violating the property, but given a value function on $L$, it can do more than that and try to direct the run towards  not only satisfying $L$, but doing it better, where \emph{better} means obtaining a larger value according to a given value function $\mu$. 
In the next sections we discuss runtime consultants for value functions for common limit operators.  
We further discuss the runtime consultant for a new value function that captures the notion of \emph{robustness}~\cite{FismanS25}, so larger will coincide with being more robust. 
Since we would like to apply the notion of a runtime consultant to both quantitative and qualitative properties, we define them wrt a quantitative automaton that generalizes both notions. 

Since modeling all system detail is often infeasible, we assume
the runtime consultant has only a partial view of the system's behavior. Formally, we assume the alphabet of the property or value function is $\Sigmaprop$, 
the alphabet of the system is $\Sigmasys$ and 
there exists a mapping $h:\Sigmasys^*\to\Sigmaprop^*$ translating 
system executions into traces relevant to the property. 
E.g. it could be that $\Sigmasys=\Sigmaprop\cup\Sigma'$ (or $\Sigmasys=\Sigmaprop\times\Sigma'$) and $h$ projects onto letters in $\Sigmaprop$ (or to the $\Sigmaprop$ component of the letter, resp.).

Let $\aut{A}$ be a quatomaton implementing a value function $\mu_{\aut{A}}:\Sigma^\omega\to\mathbb{T}$ where $\mathbb{T}$ is some totally ordered domain. While $\mu$ assigns values to infinite words, at runtime only finite words are observed. We thus need to derive from $\aut{A}$ a value to associate with a finite word, a prefix of an infinite word.
We provide two value functions for finite words, denoted $\prefixvalue(u)$ and $\lassovalue(u)$. Intuitively, for a finite word $u$ the value function \term{$\prefixvalue(u)$} returns the maximal value any infinite extension $w$ of $u$ can take. The value function \term{$\lassovalue(u)$} restricts attention to infinite words obtained by looping on the state $\aut{A}$ arrives at when reading $u$, namely $\aut{A}(u)$.

\begin{definition}[$\prefixvalue,\lassovalue$]\label{def:prefix-value-function-for-graphs}
Let $u\in\Sigma^*$ and $\aut{A}$ a quatomaton for 
${\mu_{\aut{A}}:\Sigma^\omega\to\mathbb{T}}$. 
\begin{itemize}[nosep]
\item 
{$\mathterm{\prefixvalue(u)}=\mathrm{sup}\ \{\mu_\aut{A}(uw)~|~w{\in}\Sigma^\omega\}$} is the \term{\any-extension value} of $u$.
 \item 
{$\mathterm{\lassovalue(u)}=\mathrm{sup}\ \{\mu_\aut{A}(uv^\omega)~|~v{\in}\Sigma^+\text{, } \aut{A}(uv){=}\aut{A}(u) \}$} is the \term{cyclic-extension value}.
\end{itemize}    
\end{definition}

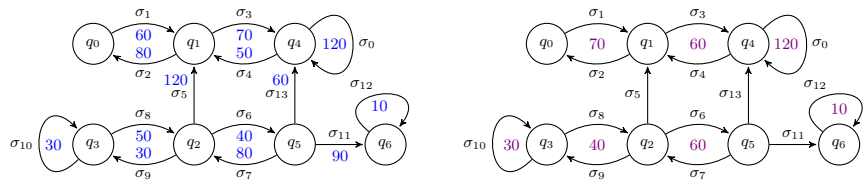
\begin{figure}[t]
\begin{center}
\scalebox{0.67}{
\begin{tikzpicture}[->,>=stealth',shorten >=1pt,auto,node distance=2cm,semithick,initial text=, initial above]

\node[state]    (q0)      {$q_0$};
\node[state]    (q1)  [right of=q0]   {$q_1$};
\node[state]    (q2)  [below of=q1,node distance=2cm]   {$q_2$};
\node[state]    (q3)  [below of=q0,node distance=2cm]   {$q_3$};
\node[state]    (q4)  [right of=q1]   {$q_4$};
\node[state]    (q5)  [right of=q2]   {$q_5$};
\node[state]    (q6)  [right of=q5,node distance=1.8cm]   {$q_6$};

\path (q4) edge [out=45, in=315, loop, looseness=7, rotate=0] 
           node {$\sigma_0$} 
           node [blue,left]{$120$} 
      (q4);
\path (q0) edge [bend left] 
           node {$\sigma_1$}
           node [blue,below]{$60$}  
      (q1);
\path (q4) edge [bend left] 
           node [below]{$\sigma_4$}
           node [blue,above]{$50$}  
      (q1);
\path (q1) edge [bend left] 
           node {$\sigma_3$}
           node [blue,below]{$70$}   
      (q4);      
\path (q1) edge [bend left] 
           node [below]{$\sigma_2$} 
           node [blue,above]{$80$} 
      (q0);
\path (q2) edge [bend left] 
           node {$\sigma_9$} 
           node [blue,above]{$30$} 
      (q3);
\path (q3) edge [bend left] 
           node [above]{$\sigma_8$} 
           node [blue,below]{$50$}  
      (q2);
\path (q2) edge  []
           node {$\sigma_5$} 
           node [blue, near end]{$120$} 
      (q1);
\path (q5) edge  []
           node []{$\sigma_{13}$}
           node [blue, near end]{$60$} 
      (q4);
\path (q3) edge [out=45, in=315, loop, looseness=10, rotate=180] 
           node {$\sigma_{10}$}
           node [blue,right]{$30$} 
      (q3);  
\path (q5) edge [bend left] 
           node {$\sigma_7$} 
           node [blue,above]{$80$} 
      (q2);
\path (q2) edge [bend left] 
           node [above]{$\sigma_6$} 
           node [blue,below]{$40$}  
      (q5);      
\path (q5) edge [] 
           node [above] {$\sigma_{11}$} 
           node [blue, below]{$90$} 
      (q6);
\path (q6) edge [out=45, in=315, loop, looseness=8, rotate=110] 
           node []{$\sigma_{12}$}
           node [blue,below]{$10$} 
      (q6);


\node[state]    (r0)   [right of=q4, node distance=5cm]   {$q_0$};
\node[state]    (r1)  [right of=r0]   {$q_1$};
\node[state]    (r2)  [below of=r1,node distance=2cm]   {$q_2$};
\node[state]    (r3)  [below of=r0,node distance=2cm]   {$q_3$};
\node[state]    (r4)  [right of=r1]   {$q_4$};
\node[state]    (r5)  [right of=r2]   {$q_5$};
\node[state]    (r6)  [right of=r5,node distance=1.8cm]   {$q_6$};
\node[label] (c12) [right of=r0, node distance=1cm] {\textcolor{violet}{$70$}};
\node[label] (c14) [right of=r1, node distance=1cm] {\textcolor{violet}{$60$}};
\node[label] (c32) [right of=r3, node distance=1cm] {\textcolor{violet}{$40$}};
\node[label] (c25) [right of=r2, node distance=1cm] {\textcolor{violet}{$60$}};
\node[label] (c6) [above of=r6, node distance=7mm] {\textcolor{violet}{$10$}};
\node[label] (c4) [right of=r4, node distance=7.5mm] {\textcolor{violet}{$120$}};
\node[label] (c3) [left of=r3, node distance=7.0mm] {\textcolor{violet}{$30$}};

\path (r4) edge [out=45, in=315, loop, looseness=7, rotate=0] 
           node {$\sigma_0$} 
      (r4);
\path (r0) edge [bend left] 
           node {$\sigma_1$}
      (r1);
\path (r4) edge [bend left] 
           node [below]{$\sigma_4$}
      (r1);
\path (r1) edge [bend left] 
           node {$\sigma_3$}
      (r4);      
\path (r1) edge [bend left] 
           node [below]{$\sigma_2$} 
      (r0);
\path (r2) edge [bend left] 
           node {$\sigma_9$} 
      (r3);
\path (r3) edge [bend left] 
           node [above]{$\sigma_8$} 
      (r2);
\path (r2) edge  []
           node {$\sigma_5$} 
      (r1);
\path (r5) edge  []
           node []{$\sigma_{13}$}
      (r4);
\path (r3) edge [out=45, in=315, loop, looseness=10, rotate=180] 
           node {$\sigma_{10}$}
      (r3);  
\path (r5) edge [bend left] 
           node {$\sigma_7$} 
      (r2);
\path (r2) edge [bend left] 
           node [above]{$\sigma_6$} 
      (r5);      
\path (r5) edge [] 
           node [above] {$\sigma_{11}$} 
      (r6);
\path (r6) edge [out=45, in=315, loop, looseness=8, rotate=110] 
           node []{$\sigma_{12}$}
      (r6);

\end{tikzpicture}}
\end{center}
\vspace{-8mm}
\caption{
The {quatomaton} for \autoref{ex:quantitative} and its simple-loops view.
}\label{fig:quantitative}
\end{figure}

\begin{example}\label{ex:quantitative}
    Consider the quatomaton $\aut{A}$ in \autoref{fig:quantitative}(left) whose  objective we call \emph{\limavg\ with the shortest path to period} --- formally, $\aut{A}$ gives a lasso word $u(v)^\omega$ the value corresponding to the average of the edges' weights on the (not necessary simple) loop $v$, minus the length of $u$. 
    For example, the values of the words $\sigma_8\sigma_5(\sigma_3\sigma_4)^\omega$ and $\sigma_4(\sigma_2\sigma_1)^\omega$ are 
    $60{-}2$ and $70{-}1$ resp., while the value of the word $\sigma_8\sigma_5(\sigma_3\sigma_4\sigma_2\sigma_1)^\omega$ is $65{-}2$. It is thus convenient to consider the graph on the right which simply visualizes the values of simple loops.

    Let us see what do $\prefixvalue(\cdot)$ and $\lassovalue(\cdot)$ return for some prefixes $u$. Note first that $\prefixvalue(u){=}120{-}x$ for any $u$ that does not have $\sigma_{11}$ or $\sigma_{12}$ as an infix, where $x$ is the length of the shortest path to $q_4$. Otherwise, $\prefixvalue(u){=}10{-}y$ where $y$ is the length of the shortest path to $q_6$ since a word that traversed $\sigma_{11}$ or $\sigma_{12}$ can only reach the final cycle $\sigma_{12}$.
    Considering $\lassovalue(\cdot)$, we observe that for any $u$ that reaches vertex $q_4$, we have $\lassovalue(u){=}120{-}x$ as $120$ is the highest value for a cycle from $q_4$. 
    For words $u$ ending in vertex $q_1$ we get $\lassovalue(u){=}80{-}z$ since the highest-value cycle starting at vertex $q_1$ is $(\sigma_3\sigma_0\sigma_4)$ and its averaged value is $80$. 
\end{example}

It is not hard to see that as we read a prefix of an infinite word, with every step the value $\prefixvalue$ can either remain unchanged or decrease (but it may not increase). The situation is different for cyclic-extensions. Here the value can increase along a prefix of an infinite word.

\begin{remark}[Monotonicity of $\prefixvalue$]\label{obs:prd-leadsto-circ-ext}
Let $\mu_\aut{A}$ be as above, and $u,u'\in\Sigma^*$.
\begin{enumerate}[nosep]
    \item \label{monoton-leadsto}
    $\prefixvalue$ is monotonically non-increasing: If $u\preceq u'$ then $\prefixvalue(u)\geq \prefixvalue(u')$.
\item  $\lassovalue$ is \underline{not} monotonic: 
Given $u\prec u'$, both $\lassovalue(u)< \lassovalue(u')$ and $\lassovalue(u)\geq \lassovalue(u')$ are possible.
\item \label{best-better-fine} $\prefixvalue(u)\geq\lassovalue(u)$ for all $u\in\Sigma^*$.
\end{enumerate}
\end{remark}

The runtime consultant should recommend which letters it is preferable to take next. We thus provide the following definitions. 

\begin{definition}[$\bestrecom(u)$]\label{def:preferred-any-extension}
Let $\mu_\aut{A}$ and $u$ be as above.
For $t\in\mathbb{T}$, a word $w\in\Sigma^\omega$ is a \term{$t$-any-extension} of $u$ if $\mu_{\aut{A}}(uw)=t$. 
If, in addition, $t=\prefixvalue(u)$ we say that $w$ is a \term{preferred any-extension} of $u$ and the set of first letters of such extension is $\mathterm{\bestrecom(u)}=\{ \sigma{\in}\Sigma ~|~ \exists w\in\Sigma^\omega \textrm { s.t. } \sigma w \textrm{ is a preferred any-extension of } u \}$.
\end{definition}

\begin{definition}[$\finerecom(u)$]\label{def:preferred-cyclic-extension}
We say that $v\in\Sigma^+$ is a cycle on $u$ if $\aut{A}(uv)=\aut{A}(u)$. 
If $v$ is a cycle on $u$ and $\mu_\aut{A}(u(v)^\omega)=t$ we say that $v$ is a \term{$t$-cyclic-extension} of $u$.
If for every cycle $v'$ of $u$ such that $v'\prec v$ we have that $v'$ is a $t'$-cyclic-extension for $t'\geq t$ we say that $v$ is a \term{$t$-preferred cyclic-extension}.\footnote{That is, if $v$ starts with a cycle $v'$ then $v'$ cannot have a value $t'$ lower than $t$.}
If, in addition, $t=\lassovalue(u)$ we say that $v$ is a \term{preferred cyclic-extension} and the set of first letters of such extension is $\mathterm{\finerecom(u)}{=}\{ \sigma{\in}\Sigma~|~\exists v{\in}\Sigma^* \textrm { s.t. } \sigma v \textrm{ is a preferred cyclic extension of } u \}$.
\end{definition}

In \autoref{ex:quantitative} the preferred any-extension of $\sigma_7$ is any word that starts with $\sigma_5$ or $\sigma_6$ and has suffix $(\sigma_0)^\omega$ since a word with such suffix has the highest achievable value. Thus $\bestrecom(\sigma_7)=\{\sigma_5,\sigma_6\}$ as otherwise the letter $\sigma_9$ is chosen, and the path to the best loop (namely, $\sigma_0$) becomes longer. 
The preferred cyclic-extension of $\sigma_7$ is $(\sigma_6\sigma_7)$, thus, $\finerecom(\sigma_7)=\{\sigma_6\}$. This is since the only cycles that include the vertex $q_2$ are $(\sigma_6\sigma_7)$, $(\sigma_8\sigma_9)$, $(\sigma_8\sigma_9\sigma_{10})$ and their combinations, and among them the cycle $(\sigma_6\sigma_7)$ has the highest value. Note that the edge $\sigma_5$ is not considered for $\finerecom(\sigma_7)$ since it takes the run to another MSCC from which vertex $q_2$ is unreachable, so no cycle can be closed.

\begin{definition}[A runtime consultant (RC)]
Let $\aut{A}$ be quatomaton, and $\arop\in\{\leadsto,\circlearrowright\}$.
A \term{$\arop$-runtime consultant for $\aut{A}$} is a procedure that given $u\in\Sigma^*$ returns a set $\Sigma'\subseteq\Sigma^{\arop}_{\aut{A}}(u)$, such that $\Sigma'$ can be empty only if $\Sigma^{\arop}_{\aut{A}}(u)$ is.
If, also, $\Sigma'=\Sigma^{\arop}_{\aut{A}}(u)$ we say that the consultant is \term{strong}. Otherwise we say it is \term{weak}.
We refer to $\Sigma^{\leadsto}_{\aut{A}}(u)$ and $\Sigma^{\circlearrowright}_{\aut{A}}(u)$ as the $\leadsto$- and $\circlearrowright$-\term{recommendations}, resp. 
\end{definition}

Ideally, an RC should work in constant time in order to keep up with the observed run and provide recommendations that can be taken in a timely manner.

\begin{example}[$\leadsto$- and $\circlearrowright$-recommendations]\label{ex:any_vs_cyclic}
  Let $\mu$ be the value function discussed in \autoref{ex:quantitative} and implemented by the quatomaton in \autoref{fig:quantitative}. Recall that the values $\mu$ assigns a lasso word depends mainly on the average weight of the loop, and the values of simple loops are given in \autoref{fig:quantitative} on the right.  

    \autoref{fig:any_vs_cyclic} represents the runtime recommendations given by the RC.
    In the left figure we annotate in blue the $\leadsto$-recommendations from each state. In the middle the $\circlearrowright$-recommendations from each state are annotated orange.
    The rightmost figure shows the paths obtained by following the $\leadsto$- and $\circlearrowright$-recommendations from the initial state, in blue and orange, respectively.

    The quatomaton  has three MSCCs. The $\leadsto$-recommendations suggest reaching as fastest as possible the best loop $(\sigma_0)$ from each state of the automaton that $\sigma_0$ is reachable from. These are  all states except the one that is reached by $\sigma_{11}$. From the latter state the best reachable cycle is $\sigma_{12}$. 
    The $\circlearrowright$-RC never recommends moving to another MSCC, because its definition requires returning to the state the run is currently at. 
    Accordingly, an edge that is not a part of any MSCC will not be recommended. Therefore, the loop $(\sigma_0)$ cannot be reached by these recommendations as long as the run is in a {bottom} MSCC. 
    
    Considering the initial state, the runs obtained by following $\leadsto$-recommen- dation are different from the run obtained following the $\circlearrowright$-recommendation.
    The $\leadsto$-recommendation reaches the MSCC with the highest-value loop and the resulting word is either $\sigma_8\sigma_5\sigma_3(\sigma_0)^\omega$ or $\sigma_8\sigma_6\sigma_{13}(\sigma_0)^\omega$. Note that there are two edges by which the run can move from the bottom MSCC to the top one and both have same path length to the best loop, so both edges are recommended by the $\leadsto$-RC. The $\circlearrowright$-recommendation stays in the initial MSCC and in every step compares the cycles that include the current state. The resulting word is $\sigma_8(\sigma_6\sigma_7)^\omega$.
\end{example}

\begin{figure}[t]
\begin{center}
\scalebox{0.8}{
\begin{tikzpicture}[->,>=stealth',shorten >=1pt,auto,node distance=1.2cm,semithick,initial text=, initial above]

\node[ministate]    (q0)    {};
\node[ministate]    (q1)  [right of=q0]  {};
\node[ministate]    (q2)  [below of=q1,node distance=1.2cm]  {};
\node[ministate]    (q3)  [below of=q0,node distance=1.2cm]  {};
\node[ministate]    (q4)  [right of=q1]  {};
\node[ministate]    (q5)  [right of=q2]  {};
\node[ministate]    (q6)  [right of=q5, node distance=0.75cm]  {};

\path (q4) edge [pblue, out=45, in=315, loop, looseness=16, rotate=0] 
           node {$\sigma_0$} 
           node [violet,left]{\tiny{$120$}} 
      (q4);
\path (q0) edge [pblue, bend left] 
           node {$\sigma_1$} 
           node [violet,below]{\tiny{$70$}} 
      (q1);
\path (q4) edge [bend left] 
           node [below, near end]{$\sigma_4$} 
      (q1);
\path (q1) edge [pblue, bend left] 
           node {$\sigma_3$} 
           node [violet,below]{\tiny{$60$}} 
      (q4);      
\path (q1) edge [bend left] 
           node [below, near end]{$\sigma_2$} 
      (q0);
\path (q2) edge [bend left] 
           node {$\sigma_9$} 
      (q3);
\path (q3) edge [pblue, bend left] 
           node [above, near start]{$\sigma_8$} 
           node [violet,below]{\tiny{$40$}} 
      (q2);
\path (q2) edge  [pblue]
           node []{$\sigma_5$} 
      (q1);
\path (q5) edge  [pblue]
           node []{$\sigma_{13}$} 
      (q4);
\path (q3) edge [out=45, in=315, loop, looseness=20, rotate=180] 
           node {$\sigma_{10}$}
           node [violet,right]{\tiny{$30$}} 
      (q3);  
\path (q5) edge [bend left] 
           node {$\sigma_7$} 
      (q2);
\path (q2) edge [pblue, bend left] 
           node [above, near start]{$\sigma_6$} 
           node [violet,below]{\tiny{$60$}} 
      (q5);      
\path (q5) edge [] 
           node [below, near start] {$\sigma_{11}$} 
      (q6);
\path (q6) edge [pblue, out=45, in=315, loop, looseness=14, rotate=110] 
            node [ above]{$\sigma_{12}$}
           node [violet,below]{\tiny{$10$}} 
      (q6);


\node[ministate]    (r0)  [right of=q4, node distance=2.7cm]  {};
\node[ministate]    (r1)  [right of=r0]  {};
\node[ministate]    (r2)  [below of=r1]  {};
\node[ministate]    (r3)  [below of=r0]  {};
\node[ministate]    (r4)  [right of=r1]  {};
\node[ministate]    (r5)  [right of=r2]  {};
\node[ministate]    (r6)  [right of=r5, node distance=0.75cm]  {};

\path (r4) edge [porange, out=45, in=315, loop, looseness=16, rotate=0] 
           node {$\sigma_0$} 
           node [violet,left]{\tiny{$120$}} 
      (r4);
\path (r0) edge [porange, bend left] 
           node {$\sigma_1$} 
           node [violet,below]{\tiny{$70$}} 
      (r1);
\path (r4) edge [bend left] 
           node [below, near end]{$\sigma_4$} 
      (r1);
\path (r1) edge [porange, bend left] 
           node {$\sigma_3$} 
           node [violet,below]{\tiny{$60$}} 
      (r4);      
\path (r1) edge [bend left] 
           node [below, near end]{$\sigma_2$} 
      (r0);
\path (r2) edge [bend left] 
           node {$\sigma_9$} 
      (r3);
\path (r3) edge [porange, bend left] 
           node [above, near start]{$\sigma_8$} 
           node [violet,below]{\tiny{$40$}} 
      (r2);
\path (r2) edge  []
           node []{$\sigma_5$} 
      (r1);
\path (r5) edge  []
           node {$\sigma_{13}$} 
      (r4);
\path (r3) edge [out=45, in=315, loop, looseness=20, rotate=180] 
           node [left]{$\sigma_{10}$}
           node [violet,right]{\tiny{$30$}} 
      (r3);  
\path (r5) edge [porange, bend left] 
           node {$\sigma_7$} 
      (r2);
\path (r2) edge [porange, bend left] 
           node [above, near start]{$\sigma_6$} 
           node [violet,below]{\tiny{$60$}} 
      (r5);      
\path (r5) edge [] 
           node [below,near start] {$\sigma_{11}$} 
      (r6);
\path (r6) edge [porange, out=45, in=315, loop, looseness=14, rotate=110] 
            node [ above]{$\sigma_{12}$}
           node [violet,below]{\tiny{$10$}} 
      (r6);


\node[ministate]    (p0)  [right of=r4, node distance=2.7cm]  {};
\node[ministate]    (p1)  [right of=p0]  {};
\node[ministate]    (p2)  [below of=p1]  {};
\node[ministate, initial]    (p3)  [below of=p0]  {};
\node[ministate]    (p4)  [right of=p1]  {};
\node[ministate]    (p5)  [right of=p2]  {};
\node[ministate]    (p6)  [right of=p5, node distance=0.75cm]  {};

\path (p4) edge [pblue, out=45, in=315, loop, looseness=16, rotate=0] 
           node {$\sigma_0$} 
           node [violet,left]{\tiny{$120$}} 
      (p4);
\path (p0) edge [bend left] 
           node {$\sigma_1$} 
           node [violet,below]{\tiny{$70$}} 
      (p1);
\path (p4) edge [bend left] 
           node [near end, below]{$\sigma_4$} 
      (p1);
\path (p1) edge [pblue, bend left] 
           node {$\sigma_3$} 
           node [violet,below]{\tiny{$60$}} 
      (p4);      
\path (p1) edge [bend left] 
           node [below, near end]{$\sigma_2$} 
      (p0);
\path (p2) edge [bend left] 
           node {$\sigma_9$} 
      (p3);
\path (p3) edge [-,side by side={blue!25}{orange!50},bend left]  
           node {} 
      (p2);
\path (p3) edge [bend left] 
           node [above, near start]{$\sigma_8$} 
           node [violet,below]{\tiny{$40$}} 
      (p2);
\path (p2) edge  [pblue]
           node []{$\sigma_5$} 
      (p1);
\path (p5) edge  [pblue]
           node [] {$\sigma_{13}$} 
      (p4);
\path (p3) edge [out=45, in=315, loop, looseness=20, rotate=180] 
           node {$\sigma_{10}$}
           node [violet,right]{\tiny{$30$}} 
      (p3);  
\path (p5) edge [porange, bend left] 
           node {$\sigma_7$} 
      (p2);
\path (p2) edge [-,side by side={blue!25}{orange!50},bend left]  
           node {} 
      (p5);
\path (p2) edge [bend left] 
           node [above, near start]{$\sigma_6$} 
           node [violet,below]{\tiny{$60$}} 
      (p5);    
\path (p5) edge [] 
           node [below,near start] {$\sigma_{11}$} 
      (p6);
\path (p6) edge [out=45, in=315, loop, looseness=14, rotate=110] 
            node [ above]{$\sigma_{12}$}
           node [violet,below]{\tiny{$10$}} 
      (p6);
\end{tikzpicture}}
\end{center}
\vspace{-6mm}
\caption{
The $\leadsto$- and $\circlearrowright$-recommendations (in blue and orange resp.)
for $\mu$ of \autoref{ex:any_vs_cyclic}.}
\label{fig:any_vs_cyclic}
\end{figure}

\section{Properties of runtime consultants}
Recall that an RC does not assume its recommendations are always followed. Below we articulate the guarantees that can be provided on \emph{partially} obeying the $\leadsto$- and $\circlearrowright$-recommendations. 
We discuss first the qualitative case, namely when the value function is Boolean, and then turn to the quantitative case.

\subsection{Properties of  RCs for Boolean value function}

Before providing a proposition 
we exemplify what happens in the Boolean case when the $\leadsto$- and $\circlearrowright$-recommendations are followed.

\begin{example}[A qualitative value function]\label{ex:qualitative}
Consider an automatic call distributor (ACD) in a business' phone system that routes calls to support staff consisting of two employees, whose functionality is 
implemented in \autoref{fig:call-distributor} by 
the DPA $\aut{A}$. Recall that a DPA accepts if the minimal rank visited infinitely often is even, thus accepting words are those that loop on the state ranked $0$ (henceforth, the $0$-state). 

The $\leadsto$-recommendation 
 for a prefix $u$ that is $\varepsilon$ (the empty word), \texttt{INIT\_SYS} or \texttt{INIT\_SYS}$\cdot$\texttt{INIT\_DB}, is $\{$\texttt{INIT\_SYS}$\}$,  $\{$\texttt{INIT\_DB}$\}$ and  $\{$\texttt{CNCT}$\}$ resp. (since otherwise the  sink state is reached, meaning a violation occurred and the word is rejected). The $\circlearrowright$-recommendation is $\emptyset$ (since there is no option to close a cycle in the automaton on the state reached by reading this prefix). 

For prefix that has not reached the sink state, $\bestrecom(u)=\Sigma$ (since an accepted word can be obtained regardless of the next read letter). 
For $u$ that reaches the $0$-state we have $\finerecom(u)=\Sigma$ as well (since any cycle starting there is accepted). 
For prefix $u$ ending with \texttt{CALL} we get 
$\finerecom(u)=\{\sigqone,\sigqtwo\}$ as these letters are part of cycles that include the $0$-state. 
For prefix $u$ ending with $\sigqone$ (resp. $\sigqtwo$) the $\circlearrowright$-recommendation is $\{$\texttt{ANS1}$\}$ ($\{$\texttt{ANS2}$\}$ resp.) as this letter must be taken in order to obtain a cycle including the $0$-state, as per \autoref{def:preferred-cyclic-extension}.
\end{example}

\begin{figure}[t]
\begin{minipage}{0.7\textwidth}
\centering
\scalebox{0.55}{
\begin{tikzpicture}[->,>=stealth',shorten >=1pt,auto,node distance=2.5cm,semithick,initial text=, initial left]

\node[state,initial]    (q0)    {$1$};
\node[state]    (q1)  [right of=q0]  {1};
\node[state]    (q2)  [right of=q1]  {1};
\node[state]    (q3)  [right of=q2]  {0};
\node[state]    (q4)  [right of=q3]  {1};
\node[state]    (q5)  [below right of=q4, node distance=2.0cm]  {1};
\node[state]    (q6)  [above right of=q4, node distance=2.0cm]  {1};
\node[state]    (q7)  [above of=q1, node distance=2.0cm]  {1};

\path (q0) edge [] 
           node {\texttt{INIT\_SYS}} 
      (q1);
\path (q0) edge [bend left=15] 
           node [left, near start]{$\Sigma\setminus$\texttt{INIT\_SYS}} 
      (q7);
\path (q1) edge [] 
           node {\texttt{INIT\_DB}} 
      (q2);
\path (q1) edge [] 
           node [above, midway,fill=white, near start]{$\Sigma\setminus$\texttt{INIT\_DB}} 
      (q7);
\path (q2) edge [] 
           node {\texttt{CNCT}} 
      (q3);
\path (q2) edge [bend right=15] 
           node [right, near start]{$\Sigma\setminus$\texttt{CNCT}} 
      (q7);
\path (q3) edge [] 
           node {\texttt{CALL}} 
      (q4);
\path (q4) edge [] 
           node [left]{$\sigqtwo$} 
      (q5);
\path (q4) edge [] 
           node [left]{$\sigqone$} 
      (q6);
\path (q4) edge [loop right] 
           node [] {$\Sigma\setminus\{\sigqone,\sigqtwo\}$}
      (q4);
\path (q6) edge [bend right] 
           node [left]{\texttt{ANS1}} 
      (q3);
\path (q5) edge [bend left] 
           node {\texttt{ANS2}} 
      (q3);  
\path (q5) edge [loop right] 
           node [] {$\Sigma\setminus$\texttt{ANS2}}
      (q5); 
\path (q6) edge [loop right] 
           node [] {$\Sigma\setminus$\texttt{ANS1}}
      (q6);
\path (q7) edge [loop right] 
           node [] {$\Sigma$}
      (q7);  
\path (q3) edge [loop below] 
           node [left, near end] {$\Sigma\setminus$\texttt{CALL}}
      (q3);  
\end{tikzpicture}}
\caption{The DPA 
for the  
ACD from \autoref{ex:qualitative}.}\label{fig:call-distributor}
\end{minipage}
\begin{minipage}{0.29\textwidth}
\centering
\scalebox{0.8}{
\begin{tikzpicture}
  
   \coordinate (P1topminus) at (2.5,0);
   \coordinate (P1botminus) at (2.5,-0.1);
   \coordinate (P1top) at (2.5,0);
   \coordinate (P1plus) at (2.5,0.1);

   \coordinate (P2topminus) at (3.95,0);
   \coordinate (P2plusminus) at (3.5,0.1);   
   \coordinate (P2top) at (4,0);
   \coordinate (P2plus) at (3.55,0.1);
   \coordinate (P22plus) at (4.4,0.15);
   
   \coordinate (P3top) at (6,.0);
   \coordinate (P4top) at (4.5,1.5);
   \coordinate (P3plus) at (4.75,.0);
   \coordinate (P4plus) at (3.3,1.3);

   \coordinate (Q1topminus) at (2.95,0);
   \coordinate (Q1botminus) at (2.95,-0.1);
   \coordinate (Q1top) at (2.5,0);
   \coordinate (Q1bot) at (2.5,-0.1);

   \coordinate (Q2topminus) at (3.95,0);
   \coordinate (Q2botminus) at (3.95,-0.1);   
   \coordinate (Q2top) at (4,0);
   \coordinate (Q2bot) at (4,-0.05);
   \coordinate (Q22bot) at (4.4,-0.15);
   
   \coordinate (Q3top) at (6,.0);
   \coordinate (Q4top) at (4.5,1.5);

   \coordinate (Q3bot) at (5,0.1);
   \coordinate (Q3botminus) at (4.95,0.1);
   \coordinate (Q4bot) at (6.75,0.95);
   \coordinate (Q5bot) at (6.25,-1);

\draw[black] (P1top)--(P2topminus) node [yshift=.275, fill=white, midway] {$u$ };
\draw[black] (P2top) .. controls (P3top) and (P4top)  .. (P2top) node [above, midway] {$v$};

\draw[gray] (Q1bot) .. controls (Q2bot) and (Q22bot)  .. (Q3botminus) node [below, midway,yshift=2] {$u''$};
\draw[gray] (Q3bot) .. controls (Q4bot) and (Q5bot)  .. (Q3bot) node [right, midway] {$v''$};

\draw[gray] (P1plus)--(P2plusminus) node [above, midway] {$u'$};
\draw[gray] (P2plus) .. controls (P3plus) and (P4plus)  .. (P2plus) node [above, midway] {$v'$};

\end{tikzpicture}}
\caption{Following $\circlearrowright$-RC as discussed in \autoref{subsec:prop-rc-general-value-function}.}\label{fig:cyclic-rc-general-value-function}
\end{minipage}
\end{figure}

\begin{restatable}[RC compliance guarantees: qualitative case]{proposition}{propmulcnstprops}\label{prop:mu-l-cnst-props}
    Let $\aut{A}_L$ 
    an automaton recognizing {a non-empty} 
    $\omega$-regular language $L$. Let $w \in\Sigma^\omega$.
    \begin{enumerate}[nosep]
        \item \label{prop:mu-l-cnst-props-first-item}
    If $w[i{+}1]\notin \Sigma_{\aut{A}_L}^\leadsto(w[..i])$ for some $i\in\mathbb{N}$ then $w\notin L$.
    \item \label{prop:mu-l-cnst-props-second-item}
    If for some $i{\in}\mathbb{N}$ and for  all $j{>}i$ we have $w[j{+}1]\notin\Sigma_{\aut{A}_L}^\circlearrowright(w[..j])$ then $w\notin L$. 
    \item \label{prop:mu-l-cnst-props-third-item}
   If for all $i{\in}\mathbb{N}$ we have that $w[i{+}1]\in\Sigma_{\aut{A}_L}^\leadsto(w[..i])$
    and for some $j{\in}\mathbb{N}$ we have that $w[k{+}1]\in\Sigma_{\aut{A}_L}^\circlearrowright(w[..k])$  for all $k\geq j$ then $w\in L$.
        \end{enumerate}
\end{restatable}

\begin{proof}
    \begin{enumerate}
        \item There are only two options for values for a word wrt the language in this setting: $\true$ and $\false$. If there exists an  $i{\in}\mathbb{N}$ such that {$w[i{+}1]\notin\Sigma_{\aut{A}_L}^\leadsto(w[..i])$} 
        then none of the actions that can achieve $\true$ are chosen. Then only value $\false$ can be achieved from now on. The $\prefixvalue$ values are non-increasing, so once the maximum value is $\false$ it will remain $\false$ all along, implying the obtained property is not in $L$.
        \item If there exists an  $i{\in}\mathbb{N}$ such that for all $j{>}i$ we have that {$w[j{+}1]\notin\Sigma_{\aut{A}_L}^\circlearrowright(w[..j])$} 
       then for each step starting from $i$, the next letter is not on a path to close a cycle with value $\true$. Therefore, a cycle with value $\true$ is never achieved, implying the cycles that will be closed have value $\false$ and the word is rejected.\footnote{From some point the word has to close cycles, since the automaton for $L$ has a finite number of states.} 
        \item If for every $i{\in}\mathbb{N}$ the next letter $w[i{+}1]$ is a letter from $\Sigma_{\aut{A}_L}^\leadsto(w[..i])$  
        then value $\true$ can be achieved. If in addition, from some  $j{\in}\mathbb{N}$ onward the next letter is a letter from  
        $\Sigma_{\aut{A}_L}^\circlearrowright(w[..j])$ 
        then a cycle of value $\true$ can be achieved and only such cycles are closed from now on. Hence, a word of value $\true$ is achieved and such words are {accepted}.\qedhere
    \end{enumerate} 
\end{proof}

We exemplify \autoref{prop:mu-l-cnst-props} on \autoref{ex:qualitative}.
Note first that deviation from $\bestrecom(\cdot)$ 
is possible only in the first three steps, since after reading \texttt{CNCT}, for any letter read there exists a continuation that leads to an accepted word.
If such a deviation occurs, the sink-state is reached and the property is violated, 
as per~\autoref{prop:mu-l-cnst-props-first-item}. 

With regard to  \autoref{prop:mu-l-cnst-props-second-item}, consider the state reached by reading \texttt{CALL}. 
If $\finerecom(\cdot)$ 
is ignored indefinitely, that is, neither $\sigqone$ nor $\sigqtwo$ are read, then the run will be stuck at a state ranked $1$ and the word will be rejected.

Last,  if $\bestrecom(\cdot)$  
is always followed then the execution does not reach the sink and the word can be accepted. If from some point onward $\finerecom(\cdot)$ 
is followed then as per \autoref{ex:qualitative}
\texttt{ANS1} or \texttt{ANS2} will be taken, hence
the $0$-state will be visited infinitely often and the word is accepted, as per \autoref{prop:mu-l-cnst-props-third-item}. 

\begin{remark}[Generalizing a runtime monitor]
    Note that \autoref{prop:mu-l-cnst-props}, \autoref{prop:mu-l-cnst-props-first-item} 
    essentially says that the 
    runtime consultant generalizes 
    a runtime monitor: if a letter that is not a $\leadsto$-recommendation of the runtime-consultant is taken, we can report violation occurred, as a traditional monitor would do.
\end{remark}

\begin{remark}[Converting an RC to an EM]
    An RC can be turned into an enforcement monitor as follows: as long as the action is in both the $\leadsto$- 
    and $\circlearrowright$-
    recommendations, it is passed as is. If it is not in the $\leadsto$-recommendation it is corrected to an action in the $\leadsto$-recommendation. At some point it should be decided that from now on all actions follow also the $\circlearrowright$-recommendation and turned into such if they are not.
\end{remark}

\begin{remark}[Limitation of RC for qualitative value function]
    Note that in the qualitative case the RC does not recommend what might intuitively be a better word, since it has a binary classification of words. A word is either accepted or rejected, and there is no distinction between accepted words. Consider a scenario similar to what happens in the ACD (\autoref{ex:qualitative}), only that for some reason answers from the first employee are only possible after another round of waiting. The RC will still recommend either $\sigqone$ or $\sigqtwo$ since both lead to accepting cycle, even though recommending $\sigqtwo$ is likely faster. For this reason we define RC not only on qualitative value functions, but also on quantitative ones. In particular, in \autoref{sec:rc-rbs-L} we discuss a quantitative interpretation of $\omega$-regular properties following~\cite{FismanS25}.
\end{remark}

\subsection{Properties of RCs for a general value function}\label{subsec:prop-rc-general-value-function}

Consider now an RC for a general value function $\mu$. 
As we assert next, 
if from some point onward all choices follow both 
$\leadsto$- and $\circlearrowright$-recommendations,
then a best word from that point is obtained. 
Furthermore, a run that follows the {$\leadsto$-recommendation} all along and from some point follows {the $\circlearrowright$-recommendation}, creates an $\omega$-word with the best possible value of $\mu$ overall.

To understand the idea of the {$\circlearrowright$-recommendation} note that if starting the periodic part after the next letter is better than starting after the current letter,  the value will improve. 
This way for a run following  the  {$\circlearrowright$-recommendations} the value
will never decrease and may improve until settling on a sufficiently good word. Put otherwise, if starting  
from some point $i$ onward  {the $\circlearrowright$-recommendation is} always followed and the word that is formed is $w=u(v)^\omega$,
then for any prefix $u'$ of $w$ of length at least $i$ there is no lasso word $u'(v')^\omega$ that obtains a better value than $u(v)^\omega$. This is true regardless if $u'$ is a prefix or an extension of $u$, as exemplified in \autoref{fig:cyclic-rc-general-value-function}.

\begin{restatable}
[RC compliance guarantees: quantitative case]{proposition}{propmugeneralcnstprops}\label{prop:mu-general-cnst-props} 
    Let $\mu$ be a value function and  $\aut{A}$ be a quatomaton implementing $\mu$. Let 
    $w\in\Sigma^\omega$.
    \begin{enumerate}[nosep]
 \item  
    Let $i,l$ such that $l{\geq} i$, $w[j{+}1]\in\bestrecom(w[..j])$ for all $j {\geq} i$, and $w[k{+}1]\in\finerecom(w[..k])$  for all $k {\geq} l$. Then 
 $\mu(w)\geq\mu(w[..i]w')$ for all $w'{\in}\Sigma^\omega$. 
\item 
    Let $i$ be such that $w[j{+}1]\in\finerecom(w[..j])$ for all $j {\geq} i$.  Let $u\prec w$ s.t. $|u|{\geq} i$ and let $v\in\Sigma^+$ close a cycle on $u$. Then 
$\mu(w)\geq \mu(u(v)^\omega)$.
\label{prop:preferred-fine}

    \end{enumerate}
\end{restatable}

\begin{proof}
    \begin{enumerate}
        \item Given $w[..i]$, the highest value that can be achieved is $\mu(w[..i]w')$ for some $w'{\in}\Sigma^\omega$. If $w[j{+}1]{\in}\bestrecom(w[..j])$ 
        for every $j$ starting from index $i$ onward, 
        then for every step $j{\geq} i$ it holds that $\prefixvalue(w[..j])=\prefixvalue(w[..j{-}1])$. By \autoref{obs:prd-leadsto-circ-ext}, $\mu^\leadsto_{\aut{A}}$ is non-increasing so the highest achievable value for each step is the same. 
        Moreover, there exists $l{\geq} i$ such that the word $w$ reaches its highest-value achievable period on the $l$-th letter of $w$. If $w[k{+}1]{\in}\finerecom(w[..k])$ 
        for every $k$ starting from index $l$ onward, 
        then by choosing a letter from $\bestrecom(w[..k])$ the highest value is achievable and, according to \autoref{prop:preferred-fine}, by choosing a letter from $\finerecom(w[..k])$ for each step the maximal value is achieved and $\mu(w)\geq\mu(w[..i]w')$ for all $w'{\in}\Sigma^\omega$.
         \item If $w[j{+}1]{\in}\finerecom(w[..j])$ from index $i$ onward then for every $j{\geq} i$ it holds that $\lassovalue(w[..j]){\geq}\lassovalue (w[..j{-}1])$ because $\aut{A}(w[..j])$ and $\aut{A}(w[..j{-}1])$ have a common cycle. That is, by reading a letter $w[j]{\in}\finerecom(w[..j{-}1])$,  
         the value for this cycle is still an option for $w[..j]$. But $\aut{A}(w[..j])$ can be a part of a better cycle so the possible value for $w[..j]$ can be higher than for $w[..j{-}1]$. That is, at each step the chosen letter is the first letter of a highest-value cycle that could be read starting from current prefix. 
         Since the number of simple cycles {and trails} is finite, 
         eventually the non-decreasing sequence of values $\lassovalue(w[..j])$ reaches its maximum 
         for a prefix of $w$ and then continues to go over this cycle. 
         Then, each cycle that is closed before reaching the cycle described above is at most as good as the chosen one. Also, every cycle that is closed after reaching the cycle described above is at most as good as this cycle, according to the maximality of its cyclic-value. Therefore, for all prefixes $u{\preceq} w$ such that $u{>}|i|$ and for all $v$ s.t.  $\aut{A}(uv){=}\aut{A}(u)$ we have $\mu(w){\geq} \mu(u(v)^\omega)$.
        \qedhere
    \end{enumerate}
\end{proof}

\subsection{A strong RC is not always possible}
Recall that the definition of a runtime consultant includes a weak and a strong version. The strong version returns the set of all letters that following them can lead to a best value for the infinite word, whereas the weak version settles with a subset of these. The motivation for the strong version is that we do not want the RC to make arbitrary choices, rather we would like it to provide the runtime controller with all the options that are equally best, and let it decide. As we show next, there are cases where it is \emph{impossible} to construct a strong RC.

\begin{restatable}[]{proposition}{propnobestrcforreach}\label{prop:no-best-rc-for-reach}
There exist instances of the reachability problem for which there is no strong $\leadsto$-RC.\footnote{A reachability objective aims to reach some target to satisfy the property.}
\end{restatable}

\begin{proof}
    Assume towards contradiction that there exists a strong $\leadsto$-RC for any reachability problem, call it \alg{R}.
    Consider a graph with two reachability targets $r_1$ and $r_2$. Assume current vertex $v$ has a path of length $k_1$ that starts with edge $e_1$ to $r_1$. Also, $v$ has a path of length $k_2$ that starts from edge $e_2$ to $r_2$. Edge $e_1$ (resp. $e_2$) connects $v$ to $v_1$ (resp. $v_2$) and $v_1$ (resp. $v_2$) has an edge back to $v$. Since both $e_1$ and $e_2$ lead to a reachability target, \alg{R} has to recommend both letters corresponding to $e_1$ and $e_2$. Assume wlog $e_1$ is chosen and the run moves to $v_1$. Then, $v_1$ has a path of length $k_1{-}1$ to $r_1$ and a path of length $k_2{+}1$ to $r_2$. Then, the recommendation of \alg{R} from $v_1$ consists of both letters corresponding to the first edges of the paths described above, as they both lead to a reachability target. Thus, an edge that returns to $v$ can be chosen. 
    Note that neither $v$ and $v_1$, nor one of the edges between them is a reachability target. 
    Thus, if the run continues indefinitely by cycling between 
$v$ and $v_1$, without ever reaching $r_1$ or $r_2$—the resulting word violates the reachability objective—contradicting the claim that following 
\alg{R}'s recommendations ensures maximal value.
\end{proof}

Intuitively, the absence of a strong $\leadsto$-RC for reachability stems from the Boolean nature of the reachability value function, that ignores how far a target is. As a result, the recommendations need not guide the system toward a closer target. 
In fact, suggesting all actions that might eventually lead to a target can result in cycles where the run loops indefinitely without ever reaching the target.

In cases where no strong RC is possible, we thus suffice with a weak RC. A weak RC for reachability targets can be implemented by recommending actions that lead to vertices whose distance from the reachability targets is smaller. In other words, replacing the Boolean reachability value function, with a quantitative value function of getting closer to a reachability target. 

As we shall see in the next section, this impossibility  implies various common quantitative operators do not have a strong $\leadsto$-RC.

For $\circlearrowright$-RC we shall see that various value functions reduce to the problem we term the \emph{cycle-reachability} problem. This is the problem of reaching from a current vertex $v$ one of the reachability targets $t$ from which $v$ is reachable. This problem has no strong $\circlearrowright$-RC.

\begin{restatable}[]{proposition}{propnofinercforreach}\label{prop:no-fine-rc-for-reach}
There exist instances of the cycle-reachability problem for which there is no strong $\circlearrowright$-RC.
\end{restatable}

The proof follows the same idea as in the proof of \autoref{prop:no-best-rc-for-reach} (see \autoref{proof:prop:no-fine-rc-for-reach}).

\section{Runtime consultants for widely used operators}\label{sec:rc-aggregators}

We turn to discuss RCs for limit operators. These operators are often used in verification to model parameters of the examined system. For example, given a weighted labeled graph that represents power usage,
the \emph{supremum} models peak power consumption. In a system that receives requests and generates responses, the \emph{limit average} value function models the average response time~\cite{ChatterjeeDH10}. 
Recall that our aim is to maximize the value wrt the examined value function.

Throughout this section we assume $\aut{A}$ is a weighed labeled graph $(\Sigma,Q,q_0,\delta,\allowbreak \val)$ where $\val$ is the considered operator, and $\theta_{max}$ is the largest weight on the graph.
We give here proof sketches, the complete proofs are in the appendix.

\paragraph{\limsup\ and \sup} 
We start with the \limsup\ and \sup\ operators. 

\begin{restatable}[RC for \limsup]{proposition}{proprcforlimsup}\label{prop:rc-for-limsup}
Let $\arop\in\{\leadsto,\circlearrowright\}$.
    \begin{enumerate}[nosep]
        \item A strong $\arop$-RC for \limsup\ does not always exist.
        \item A weak $\arop$-RC for \limsup\ can be constructed in polynomial time, such that the running time of the constructed weak $\arop$-RC is constant. 
    \end{enumerate}
\end{restatable}

\begin{proof}[Proof sketch]
 The question of constructing a 
 $\leadsto$-RC and a $\circlearrowright$-RC for $\limsup$ reduces to reachability and cycle-reachability resp. where the reachability targets are cycles with an edge of the highest weight. It follows from \autoref{prop:no-best-rc-for-reach} and \autoref{prop:no-fine-rc-for-reach}  that $\limsup$ has no strong RCs in the general case. To obtain constant running time, during preprocessing we associate with each vertex the set of letters on {its} outgoing edges that can lead to the highest weight from that point. The full proof can be found in \autoref{proof:prop:rc-for-limsup}.
\end{proof}

\begin{restatable}[RC for \sup]{proposition}{proprcforsup}\label{prop:rc-for-sup}
Let $\arop\in\{\leadsto,\circlearrowright\}$.
    \begin{enumerate}[nosep]
        \item A strong $\arop$-RC for \sup\ does not always exist.
        \item A weak $\arop$-RC for \sup\ can be constructed in polynomial time, such that the running time of the constructed weak $\arop$-RC is  $\log(\theta_{max})$. 
    \end{enumerate}
\end{restatable}
    
\begin{proof}[Proof sketch]
The $\leadsto$-RC problem for \sup\  reduces to reachability problem where the reachability targets are the edges with maximal {weight} that {are} greater than the maximal {weight} seen so far.
Accordingly, a strong RC does not always exist.
 Given the maximal weight in the graph is $\theta_{max}$, the maximal {weight} seen so far, call it $\theta_{cur}$, can be saved in {$\log(\theta_{max})$} bits. 
 To obtain $\log(\theta_{max})$ running time, during preprocessing we associate with each vertex and each weight corresponding to $\theta_{cur}$, the set of letters on outgoing edges that can lead to the highest weight above $\theta_{cur}$ from that point. If this set is empty the RC recommends $\Sigma$.
The $\circlearrowright$-RC for $\sup$ question reduces to the above reachability problem restricted to the current MSCC. 
See \autoref{proof:prop:rc-for-sup} for the full proof.
\end{proof}

\paragraph{\liminf\ and \inf}

\begin{restatable}[RC for \liminf ]{proposition}{proprcforliminf}\label{prop:rc-for-liminf}
    \begin{enumerate}[nosep]
        \item A strong $\leadsto$-RC for \liminf\ does not always exist.
        \item A weak $\leadsto$-RC and a strong $\circlearrowright$-RC 
        for \liminf\ can be constructed in polynomial time, such that the running time of the constructed RC is constant. 
    \end{enumerate}
\end{restatable}
\begin{proof}[Proof sketch]
    The $\leadsto$-RC problem for \liminf\  reduces to a reachability problem where the reachability targets are the edges of a cycle with maximal minimum-weight. Accordingly, a strong RC does not always exist.
    To obtain constant running time for the weak $\leadsto$-RC, during preprocessing we associate with each vertex the set of letters on outgoing edges that lead to a cycle with highest minimum-weight from that point. 
    For the strong $\circlearrowright$-RC the set of letters corresponds to outgoing edges that lead to a cycle with highest minimum-weight that includes the current {vertex}. 
    The full proof is available in \autoref{proof:prop:rc-for-liminf}.
\end{proof}

\begin{restatable}[RC for \inf]{proposition}{proprcforinf}\label{prop:rc-for-inf}
Let $\arop\in\{\leadsto,\circlearrowright\}$.
    \begin{enumerate}[nosep]
        \item [] A strong $\arop$-RC for \inf\ can be constructed in polynomial time, such that the running time of the constructed strong $\arop$-RC is $\log(\theta_{max})$. 
    \end{enumerate}
\end{restatable}

\begin{proof}[Proof sketch]
    Given the maximal weight in the graph is $\theta_{max}$, the minimal value seen so far, call it $\theta_{cur}$, 
    can be saved in {$\log(\theta_{max})$} bits. 
    For the $\leadsto$-RC, to obtain $\log(\theta_{max})$ running time, during preprocessing we associate with each vertex the maximal minimum-value that can be obtained, based on $\theta_{cur}$. Similarly for $\circlearrowright$-RC, but there we compute the maximal minimum-value that can be obtained by a cycle that includes the current {vertex}. The full proof can be found in \autoref{proof:prop:rc-for-inf}.
\end{proof}

\paragraph{\limavg }
\quad 
So far, RC construction for all operators 
was feasible in polynomial time. 
For \limavg\ this is not the case. That is, the $\leadsto$-RC construction can be done in polynomial time, but the construction for $\circlearrowright$-RC is more challenging. Here the problem reduces to the problem of finding a 
 cyclic-trail, where a {path} is a \emph{trail} if no edge in it appears twice. Note that vertices can appear more than once in a trail. For instance, recall the quatomaton for a value function similar to \limavg\ given in \autoref{fig:quantitative}. As discussed in~\autoref{ex:quantitative}, when in state $q_1$ the best cyclic-extension is the non-simple loop, i.e. trail,  $\sigma_3\sigma_0\sigma_4$. 
The following proposition states this reduction.

\begin{restatable}[]{proposition}{propfinetrialtoLtrail}\label{prop:fine-trial-to-Ltrail}
Finding $\circlearrowright$-recommendations for \limavg\ reduces to finding cyclic trails of maximal mean-weight involving the current vertex.
\end{restatable}

Intuitively, the reason is that 
{the} $\circlearrowright$-RC looks for the maximal mean-weight {cycle} that includes the current {vertex}, {and} if there is {no} maximal mean-weight {simple} cycle reachable from the current vertex it has to be extended to a trail that contains the path from the current vertex to the cycle and from there back to the current vertex. See full proof in~\autoref{proof:prop:fine-trial-to-Ltrail}.
This maximization problem can be turned into a minimization problem by multiplying the weights by $-1$. 
We show that the respective decision problem is coNP-complete. We define
$$L_{\trail}= \left\{ (G,k,v)~\left|~\begin{array}{l} G=(V,E,\theta) \mbox{ is a weighted graph, } v\in V,\\
    k \in\mathbb{Q}\mbox{ and \textbf{every cyclic trail }} c \mbox{ going  } \\ 
    \mbox{through } v \mbox{ has mean weight} \geq k
    \end{array}\right.
    \right\}.$$

\begin{restatable}[]{proposition}{propLtrailcoNPcomp}\label{prop:Ltrail-coNP-comp}
    $L_{\trail}$ is coNP-complete.
\end{restatable}

We prove this by a reduction from 3SAT, see \autoref{proof:prop:Ltrail-coNP-comp} and \autoref{fig:fine-trail-reduction}. The following propostition summarizes the results for \limavg.

\begin{restatable}[RC for \limavg]{proposition}{proprcforlimavg}\label{prop:rc-for-limavg}
    \begin{enumerate}[nosep]
        \item A strong $\leadsto$-RC for \limavg\ does not always exist.
        \item A weak $\leadsto$-RC for \limavg\ can be constructed in polynomial time, such that the running time of the constructed weak $\leadsto$-RC is constant.
        \item A strong $\circlearrowright$-RC for \limavg\ can be constructed in exponential time, such that the running time of the constructed strong $\circlearrowright$-RC is constant.
        \item  Computing the strong $\circlearrowright$-RC for $\limavg$ value function is coNP-complete.
    \end{enumerate}
\end{restatable}

\begin{proof}[Proof sketch]
$\leadsto$-RC reduces to 
reachability of cycles with maximal mean-weight. Thus the impossibility result for strong $\leadsto$-RC.

A weak $\leadsto$-RC can be constructed by associating with each vertex $v$ the set of letters $\Sigma_v$ that are on {a} shortest path to a cycle with the maximal mean-weight that is reachable from $v$. This association can be done in polynomial time by first finding the maximal mean-weight cycles using Karp's algorithm~\cite{Karp78}, and then propagating this information backwards.

For $\circlearrowright$-RC the problem reduces to finding the maximal mean-weight trails that include the current vertex (\autoref{prop:fine-trial-to-Ltrail}) and is coNP-complete by \autoref{prop:Ltrail-coNP-comp}. We can construct a $\circlearrowright$-RC in exponential time by considering all paths of length bounded by $|E|$ as candidates for trails. See the full proof in~\autoref{proof:prop:rc-for-limavg}.
\end{proof}

\section{Runtime consultants for $\omega$-regular properties }\label{sec:rc-omega-regular}
We turn to discuss $\omega$-regular properties, i.e. properties expressible using an $\omega$-automaton such as a deterministic parity automaton.

\subsection{RCs for the qualitative case}\label{sec:rc-chi-L}
Consider an $\omega$-regular property given by a DPA $\aut{A}$.
It induces a Boolean value function: $\true$ for accepted words and $\false$ otherwise.
Recall that a word is accepted iff the minimal rank visited infinitely often is even --- implying the existence of a cycle with minimal even rank, called an accepting cycle. 
It follows that the problem of constructing $\leadsto$- and $\circlearrowright$-RCs in this setting reduces to reachability and cycle-reachability, respectively, where the reachability targets are accepting cycles. Thus, by \autoref{prop:no-best-rc-for-reach} and \autoref{prop:no-fine-rc-for-reach}, we have to suffice with weak $\leadsto$- and $\circlearrowright$-RCs.
These can be constructed in polynomial time. The running time is constant as {desired}. The proposition follows (see proof in \autoref{proof:prop:rc-for-reg-omega-qualitative}).

\begin{restatable}[RC for an $\omega$-regular property -- qualitative case]{proposition}{proprcforregomegaqualitative}\label{prop:rc-for-reg-omega-qualitative}
Let $\aut{A}$ be a DPA inducing $\chi_{\aut{A}}:\Sigma^\omega\to\{\true,\false\}$.
Let $\arop\in\{\leadsto,\circlearrowright\}$.
    \begin{enumerate}[nosep]
        \item A strong $\arop$-RC for $\chi_{\aut{A}}$ does not always exist.
        \item A weak $\arop$-RC for $\chi_{\aut{A}}$ can be constructed in polynomial time, such that the running time of the constructed weak $\arop$-RC is constant. 
    \end{enumerate}
\end{restatable}

\subsection{From qualitative to quantitative: $\omega$-regular robustness}\label{sec:rc-rbs-L}
To obtain strong RCs for $\omega$-regular properties, we can try strengthening the value function and turning it from qualitative into quantitative. 
Such a transformation was recently proposed in~\cite{FismanS25}. The idea there is to \emph{distill} from an $\omega$-regular language $L$, 
a value function $\valrbst_L:\Sigma^\omega\to\mathbb{T}$ that given a lasso word $w$, returns a value quantifying the \emph{robustness} of $w$ wrt $L$.\footnote{The question of robustness of a system has been studied a lot in formal verification~\cite{HenzingerR00,DeWulftDMR08,JaubertR11,BouyerMS15,FiliotMRST20,MascleNSTWZ21}.
Most works assume weights are part of the input. 
In contrast,~\cite{TabuadaN16,NayakNZ22,NeiderWZ22,MuranoNZ23} 
suggest to enhance two-valued semantics of temporal logic  into a five-valued semantics, where the values reflect the robustness of a word wrt a temporal logic formula. The work~\cite{FismanS25} generalizes this idea by providing infinitely many values, and providing a semantic notion agnostic to a particular representation.
}  The set $\mathbb{T}$ is a totally ordered set from which one can infer the robustness preference relation, denoted $\gtrbst_L$. That is,  $w_1{\gtrbst_L} w_2$ if $\valrbst_L(w_1) {>}\valrbst_L(w_2)$. For example, wrt the property "$a$ should occur infinitely often", and $i{<}j$, the robustness value function determines
$b^i(a)^\omega {\gtrbst_L} b^{j}(a)^\omega$, \ 
    $(b^ia)^\omega {\gtrbst_L} (b^ja)^\omega$ \ and 
    $(ba^j)^\omega {\gtrbst_L} (ba^i)^\omega$. That is, it prefers words with a higher frequency of $a$'s in the period. 
    Further, for words with the same $a$'s frequency in the period, it prefers words with more $a$'s in the transient part.

    For lack of space we do not provide further intuitions regarding the generalization of this idea to arbitrary $\omega$-regular properties and refer the reader to~\cite{FismanS25}. We continue with providing the necessary building blocks of the definition in order to discuss later RCs for the \valrbst\ value function. While $\valrbst_L$ is a semantic notion 
that is agonistic to a particular representation of $L$,
    it is easier to explain it on 
    the \emph{robustness automaton} developed in~\cite{FismanS25}.

The robustness automaton is a canonical parity automaton for an $\omega$-regular property.\footnote{More precisely, the robustness automaton is a dual DPA, as introduced in~\cite{FismanS25}.} It can be used to compute the robustness value of a lasso word, 
which is a  tuple of three components: an acceptance bit, a period value and a spoke value. 
The period and spoke values are pairs corresponding to a value computed from the respective edges of the period and the spoke, formally defined as follows.

\subsubsection*{The robustness value function}\label{subsub:rbst-val-func}
Given the robustness DPA, we 
color its edges using the ranks of the states the edge connects. 
Given an edge $e=(v,v')$ where $v$ and $v'$ have ranks $d$ and $d'$, resp., the color of $e$ is defined as follows:
\[
\begin{array}{@{\qquad}l@{\ }l@{\qquad}l@{\  }l@{\quad}l@{\ }l}
     \text{\term{white}} & \text{if } d'=-2 &
     \text{\term{green}} & \text{if } d' \text{ is even and } d \geq d' & 
     \text{\term{yellow}} & \text{if } d < d' \phantom{--} \\
     \text{\term{black}} & \text{if } d'=-1 & 
     \text{\term{red}}  & \text{if } d'  \text{ is odd and } d \geq d'. 
\end{array}
\]
We use $\mathterm{\lcolor_w(i)}$ to denote 
the color of the edge $e$ the robustness DPA goes through when reading the $i$-the letter of $w$. 
Let $w\in\Sigma^\omega$ and consider its infix $w[j..k]$. For $i\in[j..k]$ let $\mathterm{c_i}=\lcolor_w(i)$. 
Let $\mathterm{\white} = |\{i\in[j..k]\colon c_{i} = \col{white}\}|$.
The values $\mathterm{\green}$,$\mathterm{\yellow}$,$\mathterm{\red}$,$\mathterm{\black}$ are defined similarly, wrt colors $\col{green},\col{yellow},\col{red},\col{black}$, resp. 

We use 
$\mathterm{\icolors_w(j,k)}$, 
for the tuple $(\white,\green,\yellow,\red,\black)$ providing the number of letters of each color in the infix {$w[j..k]$}. 
 The \term{score of the infix} $w[j..k]$ wrt $L$, denoted $\mathterm{\score_w(j,k)}$, is the tuple $(\white\black,\green\red)$ where $\mathterm{\white\black}=\white{-}\black$, $\mathterm{\green\red}=\green{-}\red$. Its \term{averaged score}, denoted $\mathterm{\avgscore_w(j,k)}$ is $(\frac{\white\black}{l},\frac{\green\red}{l})$ where $l=|w[j..k]|$.

Finally, let $w$ be a lasso word, and assume $w=uv^\omega$ where $u$ is the shortest prefix of $w$ on which the robustness automaton loops at, and $v$ is the corresponding loop. Let $|u|=k$ and $|v|=l$, $\tau_u=\avgscore_w(1,k)$, 
$\tau_v=\avgscore_w(k{+}1,k{+}l)$, and $a_w$ is the acceptance bit (i.e. $a_w=\true$ iff $w\in L$).
The \term{robustness value of $w$ wrt $L$}, denoted $\mathterm{\valrbst_L(w)}$ is the tuple $\mathterm{(}a_w\mathterm{,}\ \tau_v\mathterm{,}\ k(\tau_u{-}\tau_v)\mathterm{)}$. The significance of the components is from left (most-significant) to right.

\subsection{RCs for the quantitative case}

We turn to discuss the construction of RCs for $\omega$-regular robustness, namely for the value function $\valrbst_L$ induced from an $\omega$-regular property $L$.

We use $\aut{P}^{\truerobustness}_L$ for the robustness automaton of an $\omega$-regular language $L$. 
We can 
assign weights to its edges based on their color, allowing us to derive back the number of edges of each color from the total path weight. Specifically, 
edge $e$ gets weight $C^2$, $C$, $0$, $-C$, $-C^2$, resp. if it is colored \col{white}, \col{green}, \col{yellow}, \col{red}, \col{black}, resp. where $C$ is the number of edges in $\aut{P}^{\truerobustness}_L$ plus one.\label{subsub:giving-weights}

\paragraph{$\leadsto$-RC for $\valrbst_L$}

Note that $\leadsto$-RC for $\valrbst_L$ looks for the most robust word. Such a word is one that is accepted and has the most robust period and spoke. Since the value of the period is more significant, we first look for accepting cycles with best period. Then we look for the most robust spoke reaching such period.

\begin{restatable}[]{proposition}{proprbstacccycletoLsimple}\label{prop:rbst-acc-cycle-to-Lsimple}
The problem of finding the most robust \emph{accepting cycle} in a parity automaton 
reduces to the question of finding the cycle with maximal mean-weight in a weighted graph that passes via a certain state.  
\end{restatable}

The proof is  in \autoref{proof:prop:rbst-acc-cycle-to-Lsimple}.
Recall that computing the weak $\leadsto$-RC for $\limavg$ can be done in polynomial time (\autoref{prop:rc-for-limavg}). 
However, for the robustness value function we show that computing the {strong $\leadsto$-RC} is coNP-complete. Loosely, this is since for the robustness value we need not only maximal mean-weight cycle but also verifying that the cycle's minimal rank is even. 
The weights  described above for the maximum problem can be multiplied by $-1$ to make it a minimum problem. 
The decision version of the minimum problem can be stated as 
$$L_{\simple}= \left\{ (G,k,v)~\left|~\begin{array}{l} G=(V,E,\theta) \mbox{ is a weighted graph, } v\in V, k \in\mathbb{Q}\\
\mbox{and \textbf{every simple cycle }} c \mbox{ going through } v \\
\mbox{has mean weight} \geq k
\end{array}\right.
\right\}$$

\begin{restatable}[]{proposition}{propLsimplecoNPcomp}\label{prop:Lsimple-coNP-comp}
    $L_{\simple}$ is coNP-complete
\end{restatable}

The proof is by reduction from the Hamiltonian cycle problem, see \autoref{proof:prop:Lsimple-coNP-comp}.

\paragraph{$\circlearrowright$-RC for $\valrbst_L$\label{paragraph:RC-for-the-quantitative-case-circlearrowright}} \quad \\
Next we turn to the problem of computing $\circlearrowright$-RC. Here, the problem of finding $\circlearrowright$-recommendation reduces to finding accepting cyclic trails of maximal mean weight involving the current {state} $q_u$
and some {state} from a set of interest $U$ (\autoref{prop:rbst-acc-trial-to-Ltrail}). 
We show that this problem as well is coNP-complete, by a simple reduction from the respective problem of \limavg\ (\autoref{prop:Ltrail-coNP-comp-rbst}). 
The following proposition summarizes the {results} for the robustness value function.

\begin{restatable}[RC for an $\omega$-regular property -- quantitative case]{proposition}{proprcforregomegaquantitative}\label{prop:rc-for-reg-omega-quantitative}
Let $\aut{A}$ be the robustness DPA for language $L$ and $\valrbst_L:\Sigma^\omega\to\mathbb{T}$ the respective robusness value function.
Let $\arop\in\{\leadsto,\circlearrowright\}$.
    \begin{enumerate}[nosep]
        \item[] A strong $\arop$-RC for {$\valrbst_L$} can be constructed in exponential time, such that the running time of the constructed strong $\arop$-RC is constant. 
    \end{enumerate}
\end{restatable}

Using exponential time we can go over all simple cycles in the case of $\leadsto${-RC} and all cyclic trails in the case of $\circlearrowright${-RC} to find the best ones, and propagate the information backwards through the edges and {states}, so that we can associate with each {state} the set of letters corresponding to the $\leadsto$- and $\circlearrowright$-recommendations. See the full proof in \autoref{proof:prop:rc-for-reg-omega-quantitative}.

\paragraph{RC for Robustness - Example} \quad \\
We provide an example illustrating the robustness DPA for three properties and the resulting $\leadsto$- and $\circlearrowright$-recommendations of the respective strong RCs.

\begin{figure}[t]
\begin{center}
\scalebox{0.6}{
\begin{tikzpicture}[->,>=stealth',shorten >=1pt,auto,node distance=2.0cm,semithick,initial text=, initial below]


\node[state,initial]    (k0)           {$0$};
\node[state]    (k1)  [right of=k0]   {$-1$};
\node[label] (qkL) [above left of=k0, node distance=1.6cm] {$\aut{P}_{L_{\ltlG a}}^\truerobustness:$};
\node[label] (rl20) [below of=k0, node distance=4.2cm] {$\begin{array}{rl}
\Sigma^\leadsto(u){=}\{a\}\\
\Sigma^\circlearrowright(u){=}\{a\}\\
\end{array}$};

\node[state,initial]    (k00)    [below of=k0, node distance=2.5cm]       {$0$};
\node[state]    (k11)  [right of=k00]   {$-1$};

\path (k0) edge [pgreen, out=45, in=315, loop, looseness=8, rotate=180] 
           node   {$a$}
      (k0); 
\path (k0) edge [pblack, bend left] 
           node {$b$} 
      (k1);
\path (k1) edge [dotted, bend left] 
           node {$\varepsilon$} 
      (k0);

\path (k00) edge [bend left] 
           node {$b$} 
      (k11);
\path (k11) edge [dotted, bend left] 
           node {$\varepsilon$} 
      (k00);  
\path (k00) edge [|-|,side by side={blue!25}{orange!50},out=45, in=315, loop, looseness=8, rotate=180] 
           node  {}
      (k00);
\path (k00) edge [out=45, in=315, loop, looseness=8, rotate=180] 
           node  {$a$}
      (k00);      


\node[state,initial]    (q1)  [left of=k0, node distance=4cm]   {$1$};

\node[state]    (q0)     [left of=q1]     {$0$};
\node[label] (qL) [above left of=q0, node distance=1.6cm] {$\aut{P}_{L_{\infty a}}^\truerobustness:$};

\node[label] (rl10) [below of=q0, node distance=4.2cm] {$\begin{array}{rl}
\Sigma^\leadsto(u){=}\{a\}\\
\Sigma^\circlearrowright(u){=}\{a\}\\
\end{array}$};
\node[label] (rl11) [below of=q1, node distance=4.2cm] {$\begin{array}{rl}
\Sigma^\leadsto(u){=}\{a\}\\
\Sigma^\circlearrowright(u){=}\{a\}\\
\end{array}$};

\node[state,initial]    (q11)  [left of=k00, node distance=4cm]   {$1$};

\node[state]    (q00)     [left of=q11]     {$0$};

\path (q0) edge [pgreen, out=45, in=315, loop, looseness=8, rotate=180] 
           node   {$a$}
      (q0); 
\path (q0) edge [pyellow, bend left] 
           node {$b$} 
      (q1);
\path (q1) edge [pgreen, bend left] 
           node {$a$} 
      (q0);
\path (q1) edge [pred, out=225, in=135, loop, looseness=8, rotate=180] 
           node  {$b$}
      (q1);

\path (q00) edge [|-|,side by side={blue!25}{orange!50},out=45, in=315, loop, looseness=8, rotate=180] 
           node  {}
      (q00);
\path (q00) edge [out=45, in=315, loop, looseness=8, rotate=180] 
           node  {$a$}
      (q00);      
\path (q00) edge [bend left] 
           node {$b$} 
      (q11);
\path (q11) edge [|-|,side by side={blue!25}{orange!50}, bend left] 
           node  {}
      (q00);
\path (q11) edge [bend left] 
           node  {$a$}
      (q00);      
\path (q11) edge [out=225, in=135, loop, looseness=8, rotate=180] 
           node  {$b$}
      (q11);     


\node[state,initial]  [right of=k1, node distance=3cm] (p3)          {$3$};
\node[state]  (p2)  [right of=p3]   {$2$};
\node[state]  (p1)  [right of=p2]   {$1$};
\node[state]  (p0)  [right of=p1]   {$0$};
\node[label] (pL) [above left of=p3, node distance=1.6cm] {$\aut{P}_{L_{a-seq}}^\truerobustness:$};
\node[label] (rl30) [below of=p3, node distance=4.2cm] {$\begin{array}{rl}
\Sigma^\leadsto(u){=}\{a\}\\
\Sigma^\circlearrowright(u){=}\{a\}\\
\end{array}$};
\node[label] (rl31) [below of=p2, node distance=4.2cm] {$\begin{array}{rl}
\Sigma^\leadsto(u){=}\{a\}\\
\Sigma^\circlearrowright(u){=}\{b\}\\
\end{array}$};
\node[label] (rl32) [below of=p1, node distance=4.2cm] {$\begin{array}{rl}
\Sigma^\leadsto(u){=}\{a\}\\
\Sigma^\circlearrowright(u){=}\{a\}\\
\end{array}$};
\node[label] (rl33) [below of=p0, node distance=4.2cm] {$\begin{array}{rl}
\Sigma^\leadsto(u){=}\{a\}\\
\Sigma^\circlearrowright(u){=}\{a\}\\
\end{array}$};

\node[state,initial]  [right of=k11, node distance=3cm] (p33)          {$3$};
\node[state]  (p22)  [right of=p33]   {$2$};
\node[state]  (p11)  [right of=p22]   {$1$};
\node[state]  (p00)  [right of=p11]   {$0$};

\path (p3) edge  [pgreen] 
           node [below] {$a$}
      (p2); 
\path (p2) edge  [pred] 
           node [below] {$a$}
      (p1);
\path (p1) edge  [pgreen] 
           node [below] {$a$}
      (p0);
\path (p0) edge  [pgreen, out=225, in=135, loop, looseness=8, rotate=180] 
           node [right] {$a$}
      (p0);      
      
\path (p3) edge [pred, out=45, in=315, loop, looseness=8, rotate=180]  
           node  {$b$}
      (p3);   
\path (p2) edge [pyellow, bend right=25]  
           node [above, near start]  {$b$}
      (p3);   
\path (p1) edge [pyellow, bend right=35]  
           node [above, near start] {$b$}
      (p3);   
\path (p0) edge [pyellow, bend right=45]  
           node [above, near start] {$b$}
      (p3);

\path (p33) edge [|-|,side by side={blue!25}{orange!50}, bend right] 
           node  {}
      (p22);
\path (p33) edge [bend right] 
           node  {$a$}
      (p22);

\path (p22) edge  [pblue, bend right] 
           node  {$a$}
      (p11);
\path (p11) edge [|-|,side by side={blue!25}{orange!50}, bend right] 
           node  {}
      (p00);
\path (p11) edge [bend right] 
           node  {$a$}
      (p00);  
\path (p00) edge [|-|,side by side={blue!25}{orange!50}, out=225, in=135, loop, looseness=8, rotate=180] 
           node  {}
      (p00);
      
\path (p00) edge [out=225, in=135, loop, looseness=8, rotate=180] 
           node  {$a$}
      (p00);         
      
\path (p33) edge [out=45, in=315, loop, looseness=8, rotate=180]  
           node  {$b$}
      (p33);   
\path (p22) edge [porange, bend right=25]  
           node [above, near start]  {$b$}
      (p33);   
\path (p11) edge [bend right=35]  
           node [above, near start] {$b$}
      (p33);   
\path (p00) edge [bend right=45]  
           node [above, near start] {$b$}
      (p33);      
      
\end{tikzpicture}}
\end{center}
\vspace{-6mm}
\caption{Runtime recommendations for  $L_{\infty a}$, $L_{\ltlG a}$, $L_{a-seq}$, resp. from \autoref{ex:runtime-consultant-for-robustness}.
}\label{Fig:recommendations}
\end{figure}
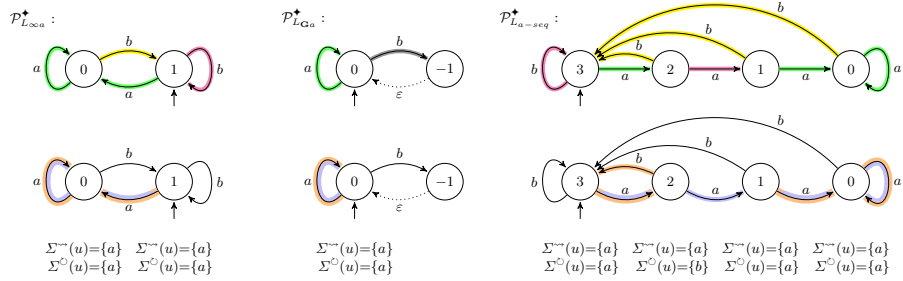

\begin{example}\label{ex:runtime-consultant-for-robustness}
Consider the following languages (i) $L_{\infty a}$ that accepts words where there are infinitely many $a$'s, (ii) $L_{\ltlG a}$ that accepts words consisting of the letter $a$ only and (iii) $L_{a-seq}=\infty a \wedge(\infty aa \to \infty aaa)$ that accepts words with infinitely many $a$'s but if there are infinitely many occurrences of $aa$ then there have to be also infinitely many occurrences of $aaa$ in the word. 
The top three automata in \autoref{Fig:recommendations} are the robustness automata for these languages. 
The bottom three automata are the same, 
annotated with $\Sigma^\leadsto(u)$ and $\Sigma^\circlearrowright(u)$
for each state. Note that in {$\aut{P}_{L_{a-seq}}^\truerobustness$} the $\leadsto$-recommendation for every state $q$ (i.e. for every finite word $u$ reaching $q$) is to read $a$ since the best period a word in $L_{a-seq}$ can have is $a^\omega$, so the {$\leadsto$-recommendation} is to reach the self-loop on the state ranked $0$. 
However, the $\circlearrowright$-recommendation is not the same for all prefixes. If the word $u\in\Sigma^*$ read so far is $a$ or has a suffix $ba$, that is, the automaton is currently in the state ranked $2$, then the best trail this state is part of is the simple cycle including it and the state ranked $3$. Note that other accepting trails from this state exist but all of them have a lower percentage of green edges and accordingly a lower value. Therefore, the $\circlearrowright$-recommendation is to read $b$ from this state. 
\end{example}

\section{Conclusions}\label{sec:discuss}
We have introduced the notion of runtime consultant (RC) that generalizes the notion of a runtime monitor. We distinguish between \emph{weak} and \emph{strong} runtime consultants, where the latter is guaranteed to make no arbitrary choices in its recommendations. We have shown that some value functions may not have a strong RC in the general case. We have provided the strongest RCs possible for the widely used operators \sup, \inf, \limsup, \liminf\, and \limavg. 

In the qualitative case we considered $\omega$-regular properties. We have shown that in general they do not have strong RCs. However, taking the recently proposed robustness value function $\valrbst_L$, that distills from a qualitative $\omega$-regular property $L$ a quantitative value function, it is possible to construct strong RCs.

The construction of most of the RCs is polynomial, except for \limavg, and $\valrbst_L$, for which we show the problem is coNP-complete. Since the construction of an RC is done at preprocessing this is not that bad.
The important measure is the running time of the RC. We show that for all cases but \sup\ and \inf\ the running time is constant. For these two it is logarithmic in the size of the heaviest weight which may also be acceptable in many settings.

\bibliographystyle{plainurl}
\bibliography{bib.bib}

\begin{thebibliography}{10}

\bibitem{BarringerGHS04}
H.~Barringer, A.~Goldberg, K.~Havelund, and K.~Sen.
\newblock Rule-based runtime verification.
\newblock In Bernhard Steffen and Giorgio Levi, editors, {\em Verification, Model Checking, and Abstract Interpretation, 5th International Conference, {VMCAI}}, volume 2937 of {\em Lecture Notes in Computer Science}, pages 44--57. Springer, 2004.

\bibitem{BartocciDDFMNS18}
E.~Bartocci, J.V. Deshmukh, A.~Donz{\'{e}}, G.~Fainekos, O.~Maler, D.~Nickovic, and S.~Sankaranarayanan.
\newblock Specification-based monitoring of cyber-physical systems: {A} survey on theory, tools and applications.
\newblock In {\em Lectures on Runtime Verification - Introductory and Advanced Topics}, volume 10457 of {\em Lecture Notes in Computer Science}, pages 135--175. Springer, 2018.

\bibitem{BouyerMS15}
P.~Bouyer, N.~Markey, and O.~Sankur.
\newblock Robust reachability in timed automata and games: {A} game-based approach.
\newblock {\em Theor. Comput. Sci.}, 563:43--74, 2015.

\bibitem{ChatterjeeDH10}
K.~Chatterjee, L.~Doyen, and T.~A. Henzinger.
\newblock Quantitative languages.
\newblock {\em {ACM} Trans. Comput. Log.}, 11(4):23:1--23:38, 2010.

\bibitem{DeWulftDMR08}
M.~De~Wulf, L.~Doyen, N.~Markey, and J-F Raskin.
\newblock Robust safety of timed automata.
\newblock {\em Formal Methods in System Design}, 33(1):45--84, 2008.

\bibitem{DolzhenkoLR15}
J.~Ligatti E.~Dolzhenko and S.~Reddy.
\newblock Modeling runtime enforcement with mandatory results automata.
\newblock {\em Int. J. Inf. Secur.}, 14(1):47–60, February 2015.
\newblock \href {https://doi.org/10.1007/s10207-014-0239-8} {\path{doi:10.1007/s10207-014-0239-8}}.

\bibitem{HubletLBKT24}
D.~Basin F.~Hublet, L.~Lima, S.~Krsti{\'{c}}, and D.~Traytel.
\newblock Proactive real-time first-order enforcement.
\newblock In A.~Gurfinkel and V.~Ganesh, editors, {\em Computer Aided Verification}, pages 156--181, Cham, 2024. Springer Nature Switzerland.

\bibitem{FalconeFM08}
Y.~Falcone, J{-}C. Fernandez, and L.~Mounier.
\newblock Synthesizing enforcement monitors wrt. the safety-progress classification of properties.
\newblock In {\em Information Systems Security, 4th International Conference, {ICISS}}, volume 5352 of {\em Lecture Notes in Computer Science}, pages 41--55. Springer, 2008.
\newblock \href {https://doi.org/10.1007/978-3-540-89862-7\_3} {\path{doi:10.1007/978-3-540-89862-7\_3}}.

\bibitem{FalconeHR13}
Y.~Falcone, K.~Havelund, and G.~Reger.
\newblock A tutorial on runtime verification.
\newblock In Manfred Broy, Doron~A. Peled, and Georg Kalus, editors, {\em Engineering Dependable Software Systems}, volume~34 of {\em {NATO} Science for Peace and Security Series, {D:} Information and Communication Security}, pages 141--175. {IOS} Press, 2013.

\bibitem{FiliotMRST20}
E.~Filiot, N.~Mazzocchi, J{-}F Raskin, S.~Sankaranarayanan, and A.~Trivedi.
\newblock Weighted transducers for robustness verification.
\newblock In {\em 31st International Conference on Concurrency Theory, {CONCUR} 2020, September 1-4, 2020, Vienna, Austria (Virtual Conference)}, pages 17:1--17:21, 2020.

\bibitem{FismanS25}
D.~Fisman and E.~Sudit.
\newblock Omega-regular robustness, 2025.
\newblock URL: \url{https://arxiv.org/abs/2503.12631}, \href {https://arxiv.org/abs/2503.12631} {\path{arXiv:2503.12631}}.

\bibitem{HavelundP23}
K.~Havelund and D.~Peled.
\newblock Monitorability for runtime verification.
\newblock In {\em Runtime Verification - 23rd International Conference, {RV} 2023, Thessaloniki, Greece, October 3-6, 2023, Proceedings}, pages 447--460, 2023.

\bibitem{HenzingerR00}
T.~A Henzinger and J-F Raskin.
\newblock Robust undecidability of timed and hybrid systems.
\newblock In {\em International Workshop on Hybrid Systems: Computation and Control}, pages 145--159. Springer, 2000.

\bibitem{JaubertR11}
R.~Jaubert and P.A.Reynier.
\newblock Quantitative robustness analysis of flat timed automata.
\newblock In {\em International Conference on Foundations of Software Science and Computational Structures}, pages 229--244. Springer, 2011.

\bibitem{johnson75}
D.~B. Johnson.
\newblock Finding all the elementary circuits of a directed graph.
\newblock {\em SIAM Journal on Computing}, 4(1):77--84, 1975.

\bibitem{Karp78}
R.~M. Karp.
\newblock A characterization of the minimum cycle mean in a digraph.
\newblock {\em Discrete Mathematics}, 23(3):309--311, 1978.

\bibitem{lee1999runtime}
Insup Lee, Sampath Kannan, Moonzoo Kim, Oleg Sokolsky, and Mahesh Viswanathan.
\newblock Runtime assurance based on formal specifications.
\newblock In {\em International Conference on Parallel and Distributed Processing Techniques and Applications}, pages 279--287. PDPTA', 1999.

\bibitem{LigattiR10}
J.~Ligatti and S.~Reddy.
\newblock A theory of runtime enforcement, with results.
\newblock In B.~Preneel D.~Gritzalis and M.~Theoharidou, editors, {\em Computer Security -- ESORICS 2010}, pages 87--100, Berlin, Heidelberg, 2010. Springer Berlin Heidelberg.

\bibitem{Maler16}
Oded Maler.
\newblock Some thoughts on runtime verification.
\newblock In Yli{\`{e}}s Falcone and C{\'{e}}sar S{\'{a}}nchez, editors, {\em Runtime Verification - 16th International Conference, {RV} 2016, Madrid, Spain, September 23-30, 2016, Proceedings}, volume 10012 of {\em Lecture Notes in Computer Science}, pages 3--14. Springer, 2016.

\bibitem{MascleNSTWZ21}
C.~Mascle, D.~Neider, M.~Schwenger, P.~Tabuada, A.~Weinert, and M.~Zimmermann.
\newblock From {LTL} to {rLTL} monitoring: improved monitorability through robust semantics.
\newblock {\em Formal Methods Syst. Des.}, 59(1):170--204, 2021.

\bibitem{MuranoNZ23}
A.~Murano, D.~Neider, and M.~Zimmermann.
\newblock Robust alternating-time temporal logic.
\newblock In {\em Logics in Artificial Intelligence - 18th European Conference, {JELIA} 2023, Dresden, Germany, September 20-22, 2023, Proceedings}, pages 796--813, 2023.

\bibitem{NayakNZ22}
S.~P. Nayak, D.~Neider, and M.~Zimmermann.
\newblock Robustness-by-construction synthesis: Adapting to the environment at runtime.
\newblock In {\em Leveraging Applications of Formal Methods, Verification and Validation. Verification Principles - 11th International Symposium, ISoLA 2022, Rhodes, Greece, October 22-30, 2022, Proceedings, Part {I}}, pages 149--173, 2022.

\bibitem{NeiderWZ22}
D.~Neider, A.~Weinert, and M.~Zimmermann.
\newblock Robust, expressive, and quantitative linear temporal logics: Pick any two for free.
\newblock {\em Inf. Comput.}, 285(Part):104810, 2022.

\bibitem{OmerP23}
M.~Omer and D.~Peled.
\newblock {\em Runtime Verification Prediction for Traces with Data}, pages 148--167.
\newblock 10 2023.
\newblock \href {https://doi.org/10.1007/978-3-031-44267-4_8} {\path{doi:10.1007/978-3-031-44267-4_8}}.

\bibitem{PerrinPinBook}
D.~Perrin and J-E. Pin.
\newblock {\em Infinite Words: Automata, Semigroups, Logic and Games}.
\newblock Elsevier, 2004.

\bibitem{BloemKKW15}
R.~K{\"o}nighofer R.~Bloem, B.~K{\"o}nighofer and C.~Wang.
\newblock Shield synthesis: Runtime enforcement for reactive systems.
\newblock In Christel Baier and Cesare Tinelli, editors, {\em Tools and Algorithms for the Construction and Analysis of Systems}, pages 533--548, Berlin, Heidelberg, 2015. Springer Berlin Heidelberg.

\bibitem{LanotteMM20}
M.~Merro R.~Lanotte and A.~Munteanu.
\newblock Runtime enforcement for control system security.
\newblock In {\em 2020 IEEE 33rd Computer Security Foundations Symposium (CSF)}, pages 246--261, 2020.
\newblock \href {https://doi.org/10.1109/CSF49147.2020.00025} {\path{doi:10.1109/CSF49147.2020.00025}}.

\bibitem{PinisettyPTJFM16}
S.~Tripakis S.~Pinisetty, V.~Preoteasa, T.~J\'{e}ron, Y.~Falcone, and H.~Marchand.
\newblock Predictive runtime enforcement.
\newblock SAC '16, page 1628–1633, New York, NY, USA, 2016. Association for Computing Machinery.
\newblock \href {https://doi.org/10.1145/2851613.2851827} {\path{doi:10.1145/2851613.2851827}}.

\bibitem{Schneider00}
F.~B. Schneider.
\newblock Enforceable security policies.
\newblock {\em ACM Trans. Inf. Syst. Secur.}, 3(1):30–50, February 2000.
\newblock \href {https://doi.org/10.1145/353323.353382} {\path{doi:10.1145/353323.353382}}.

\bibitem{TabuadaN16}
P.~Tabuada and D.~Neider.
\newblock Robust linear temporal logic.
\newblock In {\em 25th {EACSL} Annual Conference on Computer Science Logic, {CSL} 2016, August 29 - September 1, 2016, Marseille, France}, pages 10:1--10:21, 2016.

\bibitem{FalconeMFR11}
J.~C.~Fernandez Y.~Falcone, L.~Mounier and J.~L. Richier.
\newblock Runtime enforcement monitors: composition, synthesis, and enforcement abilities.
\newblock {\em Formal Methods in System Design}, 38, 06 2011.
\newblock \href {https://doi.org/10.1007/s10703-011-0114-4} {\path{doi:10.1007/s10703-011-0114-4}}.

\bibitem{ZhangMPL11}
P.~Zhang, H.~Muccini, A.~Polini, and X.~Li.
\newblock Run-time systems failure prediction via proactive monitoring.
\newblock pages 484--487, 11 2011.
\newblock \href {https://doi.org/10.1109/ASE.2011.6100105} {\path{doi:10.1109/ASE.2011.6100105}}.

\bibitem{ZhaoHFDL24}
Y.~Zhao, B.~Hoxha, G.~Fainekos, J.V. Deshmukh, and L.Lindemann.
\newblock Robust conformal prediction for {STL} runtime verification under distribution shift.
\newblock In {\em 15th {ACM/IEEE} International Conference on Cyber-Physical Systems, {ICCPS} 2024, Hong Kong, May 13-16, 2024}, pages 169--179. {IEEE}, 2024.

\end{thebibliography}

\appendix

\section{Omitted proofs}

\subsection{Proof of \autoref{prop:no-fine-rc-for-reach}}\label{proof:prop:no-fine-rc-for-reach}
\propnofinercforreach*
\begin{proof}
    Assume towards contradiction  there exists  strong $\circlearrowright$-RC for the cycle-reachability problem.
    Consider a weighted graph with one MSCC and four vertices $v_0,...,v_3$ with edges $(v_i,v_{i\oplus 1})$ and $(v_{i\oplus 1},v_i)$ for $i=[0..3]$  where $\oplus$ stands for addition modulo $4$.
    Note that the edges between consecutive vertices imply a path between each pair of vertices. Let the reachability target be the edge  $(v_0,v_1)$. Note that for any vertex there is a cycle that includes the vertex and the reachability target. Thus the strong $\circlearrowright$-RC has to recommend all outgoing edges from the current vertex.
    Assume current vertex of the graph is $v_2$. Then it has paths $v_2,v_1,v_0$ and $v_2,v_3,v_0$ to the desired edge. Since both $v_1$ and $v_3$ lead to the reachability target, the runtime consultant has to recommend the letters corresponding to both $(v_2,v_1)$ and $(v_2,v_3)$. Assume wlog $(v_2,v_1)$ is chosen and the run moves to $v_1$. 
    Then $v_1$ has paths $v_1,v_0$ and $v_1,v_2,v_3,v_0$ to the desired edge. The recommendation of the RC from $v_1$ consists of letters corresponding to both $(v_1,v_0)$ and $(v_1,v_2)$, as they both lead to the reachability target. Thus, an edge that returns to $v_2$ can be chosen. 
    Note that both edges $(v_2,v_1)$ and $(v_1,v_2)$ are not a reachability target. That is, if the run cycles between them till infinity the obtained word would not reach the reachability target, contradicting that following the recommendation gives the maximal value.  
\end{proof}

\subsection{Proof of \autoref{prop:rc-for-limsup}}\label{proof:prop:rc-for-limsup}
\proprcforlimsup*
\begin{proof}
    The $\leadsto$-RC has to recommend letters on a path to an edge with the highest reachable weight that is a part of an MSCC. This ensures that this edge can be visited infinitely often. 
    Thus, this is a reachability problem. Hence, by \autoref{prop:no-best-rc-for-reach},  a strong $\leadsto$-RC is not always possible.

    The $\circlearrowright$-RC has to recommend letters that correspond to a cycle (not necessarily a simple one) with a heaviest possible edge that includes the current vertex. 
    This means reaching the heaviest weight-edge residing on a cycle that includes the current vertex. 
    Thus, this is a cycle-reachability problem. Hence, by \autoref{prop:no-fine-rc-for-reach},  a strong $\circlearrowright$-RC is not always possible.

    The  weak $\leadsto$-RC  is constructed as follows. 
    For each MSCC $S$ of $\aut{A}$ let $M_S$ be the set of all edges with maximal weight in $S$.
    Let $M$ be their union\footnote{Recall that we use MSCC for a set of vertices that have a non-empty path between each pair. Hence a singleton $\{v\}$ is an MSCC iff $v$ has a self-loop. Accordingly, there may be {vertices} that reside in no MSCC.}.
    Then we apply an algorithm that finds the shortest path from each vertex to the heaviest reachable edge of $M$. 
    With each vertex $v$ we associate a set $\Sigma_v$ that consists of letters on the first edge of such a path. 
    At runtime, when on vertex $v$, the $\leadsto$-RC recommends $\Sigma_v$.

    The weak $\circlearrowright$-RC  is constructed as follows. 
    For vertices that are not part of any SCC the set of letters is set to the empty set, as they cannot be part of any period of any word. 
    For each vertex within an MSCC we compute the shortest path to the heaviest edge in that MSCC, and record the letters on the first edge of such a path as $\Sigma_v$. At runtime, when on vertex $v$, the $\circlearrowright$-RC recommends $\Sigma_v$.
\end{proof}

\subsection{Proof of \autoref{prop:rc-for-sup}}\label{proof:prop:rc-for-sup}
\proprcforsup*
\begin{proof}
    The $\leadsto$-RC has to recommend letters that build a path to the edge with the highest reachable weight. 
    Thus, this is a reachability problem. It follows from~\autoref{prop:no-best-rc-for-reach} that strong $\leadsto$-RCs are not possible in the general case.

    The $\circlearrowright$-RC has to recommend letters that correspond to a cycle (not necessarily a simple one) with a heaviest possible edge that includes current vertex. 
    This means reaching the heaviest weight-edge residing on a cycle that includes the current vertex. 
    Thus, this is a cycle-reachability problem, and by \autoref{prop:no-fine-rc-for-reach}  strong $\circlearrowright$-RCs are not possible in general. 

    The  weak $\leadsto$-RC  is constructed as follows. 
    For each MSCC $S$ of $\aut{A}$ let $M_S$ be the set of all edges with maximal weight in $S$.
    Let $M$ be their union, to which we add all edges that are not a part of any MSCC. 
    Then we apply an algorithm that finds the shortest path from each vertex to the heaviest reachable edge of $M$. 
    With each vertex $v$ we associate a set $\Sigma_v$ that consists of letters on the first edge of such a path, and the highest weight reachable, denoted $\theta_v$. 
    At runtime the RC traverses the graph, and updates the highest weight visited so far, call it {$\theta_{cur}$}.
    When on vertex $v$, if {$\theta_{cur}<\theta_v$} the RC recommends $\Sigma_v$. Otherwise, it recommends $\Sigma$, since the maximal weight edge was already visited, and so the rest is immaterial.

    The weak $\circlearrowright$-RC  is constructed as follows. 
    For vertices that are not part of any SCC the set of letters is set to the empty set, as they cannot be part of any period of any word. 
    For each vertex within an MSCC we compute the shortest path to the heaviest edge in that MSCC, and record the letters on the first edge of such a path as $\Sigma_v$, as well the highest weight within the MSCC, denoted $\theta_v$. 
    At runtime the RC traverses the graph, and updates the highest weight visited so far, call it {$\theta_{cur}$}.
    When on vertex $v$, if {$\theta_{cur}<\theta_v$} the RC recommends $\Sigma_v$. Otherwise, it recommends $\Sigma'\subseteq \Sigma$ which includes all letters on edges from the vertex that are within the MSCC, from the same rational as in the $\leadsto$-case.
\end{proof}

\subsection{Proof of \autoref{prop:rc-for-liminf}}\label{proof:prop:rc-for-liminf}
\proprcforliminf*
\begin{proof}
    The $\leadsto$-RC has to recommend letters that lead to a cycle with a range of weights where the minimum weight is the maximal achievable (among other cycles) and then stay within it. Thus, this is a reachability problem, and by \autoref{prop:no-best-rc-for-reach}, strong $\leadsto$-RC may not exist.  

    The weak $\leadsto$-RC  is constructed as follows. 
    For each MSCC $S$ of $\aut{A}$ let $M_S$ be the set of all edges of a cycle in $S$ with the highest minimum-weight. 
     Let $\theta_S$ be the maximal minimum-weight of a cycle in $S$. For every edge $e$ in $M_S$ we record in $\theta_e$ the weight $\theta_S$. 
    Let $M$ be the union of the $M_S$ sets.
    Then we apply an algorithm that finds the shortest path from each vertex to a reachable edge $e$ of $M$ with highest $\theta_{e}$. 
    With each vertex $v$ we associate a set $\Sigma_v$ that consists of letters on the first edge of such a path.
    At runtime, when on vertex $v$, the $\leadsto$-RC recommends $\Sigma_v$.

    The $\circlearrowright$-RC  has to recommend all the letters that are part of the highest minimum-weight cycle that includes current vertex. That is, we aim to find a path from current vertex that returns to the current vertex and its minimal weight is the highest among all other such cyclic paths.
    
    The strong $\circlearrowright$-RC  is constructed as follows. For every vertex $v$ that is inside an MSCC we find the cycles with the maximal minimum-weight that include $v$ and record the letters of their first edges as $\Sigma_{v}$.
    For vertices outside any MSCC we set $\Sigma_{v}$ to be the empty set.
    At runtime, when on vertex $v$, the $\circlearrowright$-RC recommends $\Sigma_v$.
\end{proof}

\subsection{Proof of \autoref{prop:rc-for-inf}}\label{proof:prop:rc-for-inf}
\proprcforinf*
\begin{proof}
    The $\leadsto$-RC has to recommend letters that start a path whose minimum weight is the highest compared to other paths, and are above the minimum weight seen so far, if possible. That is, this is a safety problem, asking not to visit edges below the minimum seen so far (and if impossible, lowering it as little as possible).  

    The strong $\leadsto$-RC  is constructed as follows. 
    We iteratively update the values stored for each vertex and edge as follows.
    Initially the value $\theta_e$ of the edge $e$ is its weight and the value $\theta_v$ of a vertex $v$ is the maximum of $\theta_{e_i}$ such that $e_i$ is an outgoing edge of $v$.
    In the next iteration the value $\theta_e$ of edge $e=(u,v)$ is updated to the minium between $\theta_e$ and $\theta_v$. The value $\theta_v$ of vertex $v$ is updated again as described above. This process continues until stabilization (which is bounded by $|V|$ iterations). 
    After stabilization, for each vertex $v$ and value $\theta$ we store the letters corresponding to  outgoing edges $e$ with $\theta_e$ higher or equal to $\theta$ in $\Sigma_{v,\theta}^{max}$ and the letters corresponding to outgoing  edges with the highest value among the values that are lower than $\theta$ in $\Sigma_{v,\theta}^{min}$.
    At runtime the RC traverses the graph, and updates the lowest weight visited so far, call it $\theta_{cur}$. 
    When on vertex $v$, if $\Sigma_{v,\theta_{cur}}^{max}$ is non-empty it recommends it. Otherwise, it recommends $\Sigma_{v,\theta_{cur}}^{min}$.

    For the strong $\circlearrowright$-RC we also associate with each vertex $v$ and value $\theta$ the set of letters  $\Sigma_{v,\theta}^{max}$ and $\Sigma_{v,\theta}^{min}$ to be used by the RC at runtime as above.
    For their calculation, the preprocessing is different.
    For each edge $e$ we would like to find the maximal minimum-value on a cycle  that includes $e$. This can be done iteratively as follows. Suppose the weights of edges in the original graph are $\theta_1,\ldots,\theta_k$ where $\theta_1<\theta_2<\ldots<\theta_k$. 
    In the first iteration we keep in the graph only edges with weight $\theta_k$. We iteratively keep edges of lower weight. Formally, in iteration $i$ we keep in the graph only edges with weight at least $\theta_{k-i+1}$. In each iteration we look for a cycle containing $e$. 
    If $j$ is the first iteration on which a cycle containing $e$ was found, then the value of $e$ is $\theta_{k-j+1}$. 
    Once each edge has its value, we set $\Sigma_{v,\theta}^{max}$ and $\Sigma_{v,\theta}^{min}$ as above.
\end{proof}

\subsection{Proof of \autoref{prop:fine-trial-to-Ltrail}}\label{proof:prop:fine-trial-to-Ltrail}
\propfinetrialtoLtrail*

\begin{proof} 
    For $\circlearrowright$-recommendations from vertex $v$ and an already read prefix $u$, we need the recommendation to continue $u$ towards reaching a maximal mean weight cycle involving the vertex $v$. In each step of the process, if we can close a better loop on the current vertex, we prefer that. This means that we look for the best (in terms of {max} mean-weight) cycle closing a loop on {$v$}, even if it visits a vertex more than once. That is, we allow a trail, and do not restrict to a simple cycle. We do not allow the cycle to visit an edge more than once, since otherwise one can loop on a cycle over and over, making its weight the dominant part and the weight of the rest of the cycle negligible, and this will not reflect looping back to the original vertex $v$. Thus, in our terminology we look for a cyclic-trail.
\end{proof}

\subsection{Proof of \autoref{prop:Ltrail-coNP-comp}}\label{proof:prop:Ltrail-coNP-comp}
\propLtrailcoNPcomp*

\begin{proof} 
For the sake of this proof, let $L=L_{\trail}$. Then
$$\overline{L}= \left\{ (G,k,v)~\left|~\begin{array}{l} G=(V,E,\theta) \mbox{ is a weighted graph, } v\in V,\\
k \in\mathbb{Q}\mbox{ and \textbf{there exists a cyclic trail }} c \\ \mbox{going through } v \mbox{ that has mean weight} < k
\end{array}\right.
\right\}$$ 
    
We show that $\overline{L}$ is NP-complete. Membership in NP holds since a cyclic-trail $c$ answering the requirement can be provided as a witness. The length of $c$ is bounded by $|E|$, and verifying that it answers the requirement can be done in polynomial time. We can verify that its mean weight is smaller than $k$ in polynomial time by summing the weights of its edges and dividing by the length of the cyclic trail. 

To show that it is NP-complete we do a reduction from 3SAT.
Given a 3CNF formula $\varphi$ with $m$ clauses and $n$ variables, we construct a graph $G_\varphi=(V,E,\theta)$ as follows (see illustration in \autoref{fig:fine-trail-reduction}). 
The set of vertices $V$ is
$\{v\} \cup \{ x_i, t_{ij}, t'_{ij}, f_{ij}, f'_{ij} ~|~i\in[1..n],\ j\in[1..m] \} \cup \{ c_{j\ell}~|~j\in[1..m],\ \ell\in[1..d] \}$ where $d$ is {$nm{+}1$}. 
The set of vertices $E$ is 
$$\begin{array}{l}
\phantom{\cup \ }\{(v,x_1),(t'_{nm},v),(f'_{nm},v)\}\\ \cup\  
\{(x_i,t_{i1}),(x_i,f_{i1}),(t_{ij},t'_{ij}),(f_{ij},f'_{ij})~|~i\in[1..n],\ j\in[1..m]\}\\ \cup\  
\{(t'_{ij},t_{ij+1}),(f'_{ij},f_{ij+1}) ~|~i\in[1..n],\ j\in[1..m\!-\!1]\}\\\cup\ 
\{(t'_{im},x_{i+1}),(f'_{im},x_{i+1})~|~i\in[1..n\!-\!1],\ j\in[1..m]\}\\\cup\ 
\{(t_{ij},c_{j 1}),(c_{j 1},t'_{ij}),~|~i\in[1..n],\ j, \in[1..m], x_i \in C_j\}\\\cup\ 
\{(f_{ij},c_{j 1}), (c_{j 1}, f'_{ij})~|~i\in[1..n],\ j\in[1..m], \overline{x_i}\in C_j\}\\\cup\
\{(c_{j\ell},c_{j\ell+1})~|~j\in[1..m],\ \ell\in[1..d\!-\!1]\}\\\cup \
\{(c_{jd},c_{j1})~|~j\in[1..m]\}
\end{array}.$$
The weights $\theta$ are $0$ for edges in $\{(c_{j\ell},c_{j\ell+1})~|~j{\in}[1..m],\ \ell{\in}[1..d\!-\!1]\}$ or 
$\{(c_{jd},c_{j1})~|~\allowbreak j{\in}[1..m]\}$ and $1$ for the rest of the edges. 
We set $k$ to some number in the range $(\frac{x}{y},\frac{x-1}{y-1-d})$ where {$x=(2nm{+}n{+}1){+}m$} and
$y=x{+}md$.

Let $\eta\in\mathbb{B}^n$ be an assignment to the $n$ variables.
We can construct a cycle capturing this assignment as follows.
The cycle starts at $v$, then moves to $x_{1}$.
From $x_1$ if $\eta[1]=\true$ it goes to $t_{11}$, otherwise to $f_{11}$.
Once in $t_{ij}$ it goes to $t'_{ij}$ and from $t'_{ij}$ to $t_{i(j+1)}$ and so on until reaching $t'_{im}$, and similarly for $f_{ij}$.
From $t'_{im}$ or $f'_{im}$ it goes to $x_{i+1}$.
From $x_{i+1}$ to $t_{(i\!+\!1)1}$ or $f_{(i\!+\!1)1}$ depending on whether $\eta[i\!+\!1]$ is $\true$ or $\false$, and so on.
From $t'_{nm}$ or $f'_{nm}$ it goes back to $v$. The length of this cycle is {$2mn+n+1$}.

Next, we show that if $\varphi$ has a satisfying assignment $\eta\in\mathbb{B}^n$
we can extend the cycle so that it visits all the vertices $c_{j\ell}$ for $j\in[1..m]$ and $\ell\in[1..d]$.
If $\eta[i]$ is $\true$ (resp. $\false$) and $x_i$ (resp. $\bar{x_i}$) is in the $j$-th clause then the cycle instead of going from $t_{ij}$ to $t'_{ij}$ (resp. $f_{ij}$ to $f'_{ij}$) goes
from $t_{ij}$ (resp. $f_{ij}$) to $c_{j1}$ and then goes to $c_{j2}$ and so on until reaching $c_{jd}$ and looping back to $c_{j1}$ and from there to $t'_{ij}$ (resp. $f'_{ij}$).

The cycle described above is a trail (as no edge is visited more than once) and it visits all vertices $c_{j1},\ldots,c_{jd}$ for all $j\in[1..m]$. 
The length of the cycle is extended by $md+m$;  i.e., it is {$2mn+n+1+md+m=y$}. Since the weight of the edges in the part corresponding to the clauses is $0$, the weight of this cycle is {$(2mn+n+1)+m$}, namely, $x$. The mean-weight of the trail is thus $\frac{x}{y}$ which is smaller than $k$.

Note that there exists no trail in this graph going through $v$ with a better mean weight.
Consequently, if $\varphi$ has no satisfying assignment, then for every assignment there exists at least one clause that is not satisfied. Therefore, any cycle corresponding to this assignment does not go through the vertices and edges corresponding to the unsatisfied clause.  
The mean-weight of any such cycle is at least $\frac{x-1}{y-1-d}$. Therefore, it is greater than $k$. 
\end{proof}

\begin{figure}
\begin{center}
\scalebox{0.65}{
\begin{tikzpicture}[->,>=stealth',shorten >=1pt,auto,node distance=1.6cm,semithick,initial text=, initial above,
state/.style={circle, draw, minimum size=0.8cm}]

\node[state]  (x1)    {$x_1$};
\node[state]  (x2)  [right of=x1, node distance=2.35cm]   {$x_2$};
\node[label]  (dotsForX)  [right of=x2, node distance=2.35cm]   {$\circ\circ\circ$};
\node[state]  (xn)  [right of=dotsForX, node distance=2.35cm]   {$x_n$};
\node[state]  (v)  [right of=xn, node distance=2.35cm]   {$v$};

\node[state]  (t11)  [above of=x1]   {$t_{11}$};
\node[state]  (t11tag)  [above of=t11]   {$t'_{11}$};
\node[state]  (t12)  [above of=t11tag]   {$t_{12}$};
\node[label]  (dotsForT1)  [above of=t12]   {\rotatebox{90}{$\circ\circ\circ$}};
\node[state]  (t1mtag)  [above of=dotsForT1]   {$t'_{1m}$};

\node[state]  (f11)  [below of=x1]   {$f_{11}$};
\node[state]  (f11tag)  [below of=f11]   {$f'_{11}$};
\node[state]  (f12)  [below of=f11tag]   {$f_{12}$};
\node[label]  (dotsForF1)  [below of=f12]   {\rotatebox{90}{$\circ\circ\circ$}};
\node[state]  (f1mtag)  [below of=dotsForF1]   {$f'_{1m}$};

\node[state]  (t21)  [above of=x2]   {$t_{21}$};
\node[state]  (t21tag)  [above of=t21]   {$t'_{21}$};
\node[state]  (t22)  [above of=t21tag]   {$t_{22}$};
\node[label]  (dotsForT2)  [above of=t22]   {\rotatebox{90}{$\circ\circ\circ$}};
\node[state]  (t2mtag)  [above of=dotsForT2]   {$t'_{2m}$};

\node[state]  (f21)  [below of=x2]   {$f_{21}$};
\node[state]  (f21tag)  [below of=f21]   {$f'_{21}$};
\node[state]  (f22)  [below of=f21tag]   {$f_{22}$};
\node[label]  (dotsForF2)  [below of=f22]   {\rotatebox{90}{$\circ\circ\circ$}};
\node[state]  (f2mtag)  [below of=dotsForF2]   {$f'_{2m}$};

\node[state]  (tn1)  [above of=xn]   {$t_{n1}$};
\node[state]  (tn1tag)  [above of=tn1]   {$t'_{n1}$};
\node[state]  (tn2)  [above of=tn1tag]   {$t_{n2}$};
\node[label]  (dotsForTn)  [above of=tn2]   {\rotatebox{90}{$\circ\circ\circ$}};
\node[state]  (tnmtag)  [above of=dotsForTn]   {$t'_{nm}$};

\node[state]  (fn1)  [below of=xn]   {$f_{n1}$};
\node[state]  (fn1tag)  [below of=fn1]   {$f'_{n1}$};
\node[state]  (fn2)  [below of=fn1tag]   {$f_{n2}$};
\node[label]  (dotsForFn)  [below of=fn2]   {\rotatebox{90}{$\circ\circ\circ$}};
\node[state]  (fnmtag)  [below of=dotsForFn]   {$f'_{nm}$};

\node[label]  (space)  [right of=tn2, node distance=3.0cm]   {$ $};
\node[state]  (c11)  [right of=space]   {$c_{11}$};
\node[state]  (c12)  [right of=c11]   {$c_{12}$};
\node[label]  (dotsForC11)  [right of=c12]   {$\circ\circ\circ$};
\node[state]  (c1dnm)  [right of=dotsForC11]   {$c_{1d}$};

\node[state]  (c21)  [below of=c11, node distance=1.90cm]   {$c_{21}$};
\node[state]  (c22)  [right of=c21]   {$c_{22}$};
\node[label]  (dotsForC21)  [right of=c22]   {$\circ\circ\circ$};
\node[state]  (c2dnm)  [right of=dotsForC21]   {$c_{2d}$};
\node[label]  (gapForC)  [below of=c21, node distance=1.90cm]   {};
\node[label]  (dotsForC)  [right of=gapForC, node distance=2.5cm]   {\rotatebox{90}{$\circ\circ\circ$}};
\node[state]  (cm1)  [below of=gapForC, node distance=1.40cm]   {$c_{m1}$};
\node[state]  (cm2)  [right of=cm1]   {$c_{m2}$};
\node[label]  (dotsForCm1)  [right of=cm2]   {$\circ\circ\circ$};
\node[state]  (cmdnm)  [right of=dotsForCm1]   {$c_{md}$};

\path (v) edge   [bend left=15]
           node [above] {$1$}
      (x1); 

\path (x1) edge   
           node [right] {$1$}
      (t11); 
\path (t11) edge   
           node [right] {$1$}
      (t11tag); 
\path (t11tag) edge   
           node [right] {$1$}
      (t12);      
\path (t12) edge   
           node [right] {$1$}
      (dotsForT1); 
\path (dotsForT1) edge   
           node [right] {$1$}
      (t1mtag); 

\path (x1) edge   
           node [right] {$1$}
      (f11); 
\path (f11) edge   
           node [right] {$1$}
      (f11tag); 
\path (f11tag) edge   
           node [right] {$1$}
      (f12);      
\path (f12) edge   
           node [right] {$1$}
      (dotsForF1); 
\path (dotsForF1) edge   
           node [right] {$1$}
      (f1mtag); 

\path (x2) edge   
           node [right] {$1$}
      (t21); 
\path (t21) edge   
           node [right] {$1$}
      (t21tag); 
\path (t21tag) edge   
           node [right] {$1$}
      (t22);      
\path (t22) edge   
           node [right] {$1$}
      (dotsForT2); 
\path (dotsForT2) edge   
           node [right] {$1$}
      (t2mtag); 

\path (x2) edge   
           node [right] {$1$}
      (f21); 
\path (f21) edge   
           node [right] {$1$}
      (f21tag); 
\path (f21tag) edge   
           node [right] {$1$}
      (f22);      
\path (f22) edge   
           node [right] {$1$}
      (dotsForF2); 
\path (dotsForF2) edge   
           node [right] {$1$}
      (f2mtag); 

\path (xn) edge   
           node [right] {$1$}
      (tn1); 
\path (tn1) edge   
           node [right] {$1$}
      (tn1tag); 
\path (tn1tag) edge   
           node [right] {$1$}
      (tn2);      
\path (tn2) edge   
           node [right] {$1$}
      (dotsForTn); 
\path (dotsForTn) edge   
           node [right] {$1$}
      (tnmtag); 

\path (xn) edge   
           node [right] {$1$}
      (fn1); 
\path (fn1) edge   
           node [right] {$1$}
      (fn1tag); 
\path (fn1tag) edge   
           node [right] {$1$}
      (fn2);      
\path (fn2) edge   
           node [right] {$1$}
      (dotsForFn); 
\path (dotsForFn) edge   
           node [right] {$1$}
      (fnmtag);     

\path (t1mtag) edge  [out=75,in=115] 
           node [right, near start] {$1$}
      (x2); 
\path (f1mtag) edge  [out=285,in=245] 
           node [right, near start] {$1$}
      (x2); 
\path (t2mtag) edge  [out=75,in=115] 
           node [right, near start] {$1$}
      (dotsForX); 
\path (f2mtag) edge  [out=285,in=245] 
           node [right, near start] {$1$}
      (dotsForX); 
\path (tnmtag) edge  [out=75,in=115] 
           node [right, near start] {$1$}
      (v); 
\path (fnmtag) edge  [out=285,in=245] 
           node [right, near start] {$1$}
      (v); 

\path (c11) edge   
           node [above] {$0$}
      (c12); 
\path (c12) edge   
           node [above] {$0$}
      (dotsForC11); 
\path (dotsForC11) edge   
           node [above] {$0$}
      (c1dnm);      
\path (c1dnm) edge [bend left=25]  
           node [below, near start]  {$0$}
      (c11); 

\path (c21) edge   
           node [above] {$0$}
      (c22); 
\path (c22) edge   
           node [above] {$0$}
      (dotsForC21); 
\path (dotsForC21) edge   
           node [above] {$0$}
      (c2dnm);      
\path (c2dnm) edge [bend left=25]  
           node [below, near start]  {$0$}
      (c21); 

\path (cm1) edge   
           node [above] {$0$}
      (cm2); 
\path (cm2) edge   
           node [above] {$0$}
      (dotsForCm1); 
\path (dotsForCm1) edge   
           node [above] {$0$}
      (cmdnm);      
\path (cmdnm) edge [bend left=25]  
           node [below, near start]  {$0$}
      (cm1); 

\path (tn1) edge  [bend right=5] 
           node [below, near end] {$1$}
      (c11); 
\path (c11) edge [bend right=5]  
           node [above, near start] {$1$}
      (tn1tag);   
\path (dotsForT2) edge  [bend right=5]  
           node [below, near end] {$1$}
      (cm1); 
\path (cm1) edge   [bend right=5] 
           node [above, near start] {$1$}
      (t2mtag);  
\path (dotsForFn) edge [bend left=5] 
           node [above, near end] {$1$}
      (cm1); 
\path (cm1) edge  [bend left=5]
           node [above, near start] {$1$}
      (fnmtag);     

\end{tikzpicture}}
\end{center}
\vspace{-20mm}
\caption{Graph for the reduction from 3SAT to $L_{trail}$ applied to a formula of the form $(x_n\vee \ldots)  \wedge \cdots \wedge (x_2 \vee \overline{x_n} \vee \ldots )$}\label{fig:fine-trail-reduction}
\end{figure}
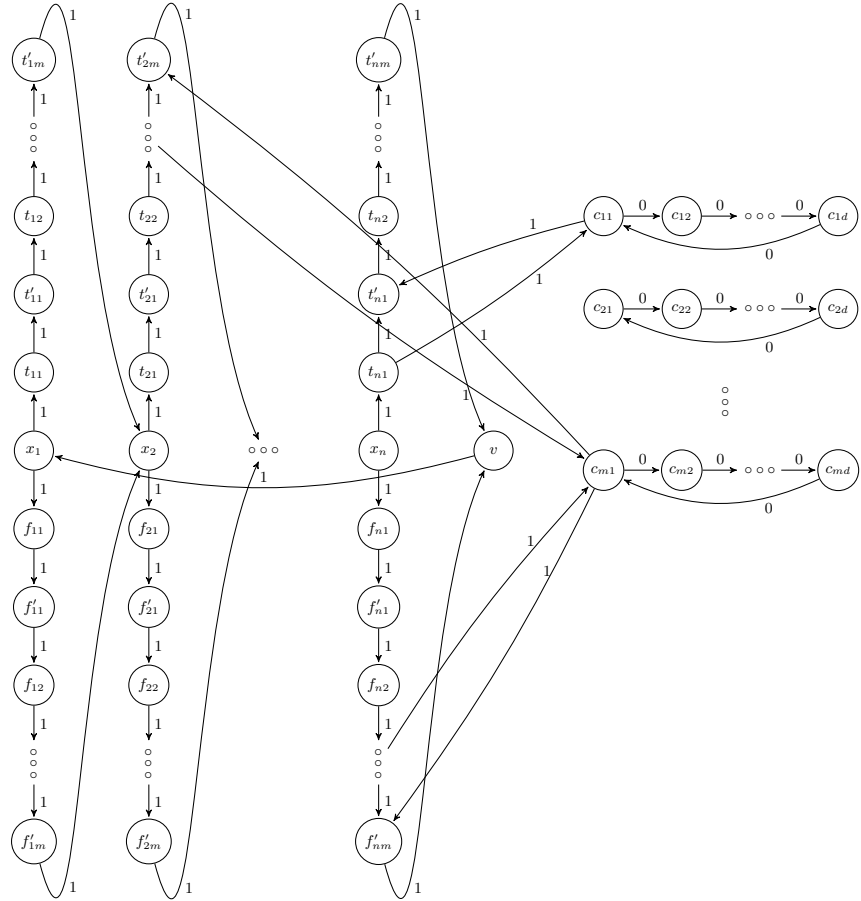

\subsection{Proof of \autoref{prop:rc-for-limavg}}\label{proof:prop:rc-for-limavg}
\proprcforlimavg*
\begin{proof}
    The $\leadsto$-RC has to recommend letters that build a path to the cycle with the highest reachable mean-weight. 
    Thus, this is a reachability problem. Hence, by~\autoref{prop:no-best-rc-for-reach}, no strong $\leadsto$-RC exists in the general case.

    The  weak $\leadsto$-RC  is constructed as follows. 
    For each MSCC $S$ of $\aut{A}$ let $M_S$ be the set of all edges of a cycle with maximal mean-weight in $S$. Such cycles can be found using Karp's algorithm~\cite{Karp78}. 
    Let $\theta_S$ be the maximal mean-weight of a cycle in $S$.
    We want to propagate these weights backwards to all edges and vertices from which $M_S$ is reachable. 
    Once this propagation is done, the weight associated with a vertex $v$ is the maximal mean-weight of a cycle reachable from $v$, call it $\theta_v$. 
    Then we apply an algorithm that finds the shortest path from each vertex $v$ to a cycle of mean-weight $\theta_v$. 
    With each vertex $v$ we associate a set $\Sigma_v$ that consists of letters on the first edge of such a path. 
    At runtime, when on vertex $v$, the $\leadsto$-RC recommends $\Sigma_v$.

    For the $\circlearrowright$-RC  we use the notion of a cyclic trail. The problem of finding the maximal mean-weight cycle that includes current vertex is reduced to the problem of finding cyclic trails of maximal mean-weight involving the current vertex, according to \autoref{prop:fine-trial-to-Ltrail}. 
    This problem is coNP-complete, as per \autoref{prop:Ltrail-coNP-comp}. 
    
    The strong $\circlearrowright$-RC  is constructed as follows. For each edge that is contained in an MSCC we compute the maximal mean-weight of a cyclic trail that the edge participates in. 
    With each vertex $v$ we associate a set $\Sigma_v$ that consists of letters on its outgoing edges with highest mean-weight.
    For vertices that are not part of any MSCC we set $\Sigma_v$ to the empty set.
    At runtime, when on vertex $v$, the $\circlearrowright$-RC recommends $\Sigma_v$.
\end{proof}

\subsection{Proof of \autoref{prop:rc-for-reg-omega-qualitative}}\label{proof:prop:rc-for-reg-omega-qualitative}

\proprcforregomegaqualitative*
\begin{proof}

    The $\leadsto$-RC has to recommend letters on a path to an accepting cycle. Thus, this is a reachability problem. Hence, by \autoref{prop:no-best-rc-for-reach},  a strong $\leadsto$-RC is not always possible.

    The $\circlearrowright$-RC has to recommend letters that correspond to an accepting cycle (not necessarily a simple one) that includes current state. Going through such cycle has to visit an even minimal rank. Thus, this is a cycle-reachability problem. Hence, by \autoref{prop:no-fine-rc-for-reach}, a strong $\circlearrowright$-RC is not always possible.

    For the weak $\leadsto$-RC for $\chi_{\aut{A}}$ the idea is to find in each MSCC some accepting cycle, and then direct the run  to an accepting cycle, by following the shortest path. 
    In an MSCC $S$ we can find an accepting cycle as follows. 
    If the minimum state-rank $d$ in $S$ is odd, then remove all states ranked $d$ and their adjacent edges. After this step, find a simple cycle that includes the minimal rank (now certainly even). This cycle is of course accepting. 
    Once an accepting simple cycle for each MSCC $S$ (if exists)  is found, we let $M$ be the set of edges on all such cycles.   
    Then we apply 
    an algorithm that finds the shortest path from each state to a reachable edge of $M$. 
    With each state $v$ we associate a set $\Sigma_v$ that consists of letters on the first edge of such a path. 
    If there are no reachable edges in $M$ from $v$, we set $\Sigma_v$ to be $\Sigma$. 
    At runtime, when on state $v$, the $\leadsto$-RC recommends $\Sigma_v$.

    The weak $\circlearrowright$-RC for $\chi_{\aut{A}}$ is constructed as follows. 
    For each MSCC $S$ of $\aut{A}$, if the minimal rank  $d$ in $S$ is odd, we temporarily remove from $S$ all states of rank $d$ and their adjacent edges.
    Then for each MSCC $S$, let $V_S$ be the set of states of the (revised) minimal rank. 
    Then we apply an algorithm that finds the shortest path from each state of $S$ to $V_S$. 
    With each state $v$ we associate a set $\Sigma_v$ that consists of letters on the first edge of such a path. 
    For states $v$ whose $\Sigma_v$ is empty after the above, we set $\Sigma_v$ to be $\Sigma$. 
    For states outside any MSCC we set the letters sets to be {the empty set}. 
    At runtime, when on state $v$, the $\circlearrowright$-RC recommends $\Sigma_v$.
\end{proof}

\subsection{Proof of \autoref{prop:rbst-acc-cycle-to-Lsimple}}\label{proof:prop:rbst-acc-cycle-to-Lsimple}
\proprbstacccycletoLsimple*
\begin{proof} 
Given the robustness parity automaton $\aut{P}^{\truerobustness}_L$  
for $L$, 
we consider the induced graph 
with weights as suggested in p.~\pageref{subsub:giving-weights}.
If a cycle is accepting then the lowest state-rank in it is even.
Therefore, for every state $s$ with even rank $d$, we can temporarily remove all states of lower rank and their incoming and outgoing edges from the graph. The state $s$ is the state the cycle should pass through. Then we find the cycle(s) with the maximal mean weight that includes the state $s$. Clearly the above weights reflect the robustness preferences. Lower weights are given to edges that should be avoided in an accepting cycle and higher weights are given to edges that relate to moves that end in a small even rank that are the preferred moves in an accepting cycle. See \autoref{rmk:leadso-rc-in-forgetful} regarding the robustness parity automaton.
\end{proof}

\begin{remark} \label{rmk:leadso-rc-in-forgetful}
    Note that 
    the above proof refers to a parity automaton, whereas the robustness automaton is a generalization of parity automaton (DPA) called a \emph{dual} parity automaton (\fdpa). The reduction holds for \fdpa\ as well. We elaborate why.
    
    The difference is that in a \fdpa\ the ranks are in $\{-2,-1,0,1,\ldots,k\}$ whereas in a DPA there are no negative ranks. Furthermore, from a state ranked $-2$ or $-1$ there {is an $\varepsilon$-transition} to the state it arrived from, and there are no further transitions.
    A run of a \fdpa\ is accepting if either (i) no-negative ranks occur, and the minimal rank visited infinitely often is even, or (ii) rank $-2$ is the first non-negative rank visited.  
    Note that by construction, all $-2$-ranked states are in a cycle with one white edge, and one transparent edge (the $\varepsilon$-edge). Furthermore, all $-2$-cycles are  best mean-weight {cycles}. Note also that such cycles are not accessible via a simple path that includes rank $-1$ since all $-1$-states are on a cycle with an edge back to the source state of the edge leading to $-1$. So a path to the $-2$-cycle will not take the $-1$-state. It follows that in spite the difference in definition of \fdpa\ and DPA the same reduction works for \fdpa\ as well.
\end{remark}

\subsection{Proof of \autoref{prop:Lsimple-coNP-comp}}\label{proof:prop:Lsimple-coNP-comp}
\propLsimplecoNPcomp*
\begin{proof}
For the sake of this proof, let $L=L_{\simple}$. Then
$$\overline{L}= \left\{ (G,k,v)~\left|~\begin{array}{l} G=(V,E,\theta) \mbox{ is a weighted graph, } v\in V, k \in\mathbb{Q}\\
\mbox{and \textbf{there exists a simple cycle }} c \mbox{ going through } v \\
\mbox{with mean weight} < k
\end{array}\right.
\right\}$$ 
    Clearly $\overline{L}$ is in NP.
    As a witness for membership of $(G,k,v)$ in $\overline{L}$ we can take a simple cycle $c$. The length of $c$ is bounded by $|V|$ and is thus polynomial. Furthermore, we can verify that its mean weight is smaller than $k$ in polynomial time by summing the weights of its edges and dividing by the length of the cycle. 
    
    Hardness of $\overline{L}$ is proven by reduction from the Hamiltonian Cycle problem. Given graph $G$, the reduction chooses one vertex $v$ arbitrarily, gives weight $1$ to its outgoing edges and weight $0$ to the rest of the edges. It sets $k$ to some number in the range $(\frac{1}{n},\frac{1}{n-1})$ where $n=|V|$. The obtained weighted graph is called $G'$.

    If there is a Hamiltonian cycle in $G$ then the corresponding cycle in the weighted graph $G'$ has weight $\frac{1}{n}$ since it must include one and only one outgoing edge from $v$. This weight is smaller than $k$.
    Any simple cycle of $G$ that includes $v$ and is not an Hamiltonian cycle, will have weight $\frac{1}{m}$ for some $m<n$.
    We get that $\frac{1}{m}\geq\frac{1}{n-1}>k$.  Therefore, if there is no Hamiltonian cycle in $G$ then the minimum mean weight of a cycle in $G'$ that includes $v$ is greater than $k$.
\end{proof}

\subsection{Proof of \autoref{prop:rbst-acc-trial-to-Ltrail}}\label{proof:prop:rbst-acc-trial-to-Ltrail}

\begin{restatable}[]{proposition}{proprbstacctrialtoLtrail}\label{prop:rbst-acc-trial-to-Ltrail}
The problem of finding $\circlearrowright$-recommendation of $\valrbst_L$ reduces to finding accepting cyclic trails of maximal mean-weight involving the current {state} $q_u$
and some {state} from a set of interest $U$. 
\end{restatable}

\begin{proof} 
    We can consider the same weights as in p.\pageref{subsub:giving-weights}.
    For $\circlearrowright$-recommendations from state $q_u$ and an already read prefix $u$, we need the recommendation to continue $u$ towards reaching a best cycle involving the state $q_u$ that contains a minimal even state-rank. In each step of the process, if we can close a better loop on the current state, we prefer that. This means that we look for the best (in terms of {max} mean-weight) accepting cycle closing a loop on {$q_u$}, even if it visits a state more than once. That is, we allow a trail, and do not restrict to a simple cycle. We do not allow the cycle to visit an edge more than once, since otherwise one can loop on a cycle over and over, making its weight the dominant part and the weight of the rest of the cycle negligible, and this will not reflect looping back to the original state $q_u$. Thus, in our terminology we look for a cyclic-trail.

    For each even rank $d$, we first remove  the {states} of rank smaller than $d$. This way intersecting a {state} ranked $d$ guarantees the resulting cycle is accepting. 
    We define $U$ to be the set of states of rank $d$. This ensures that a cycle that intersects $U$ is an accepting cycle. 
    Given $k$ is the range of ranks, we have to run the above procedure at most $k/2$ times for each {state} $q_u$ for all even ranks smaller or equal to the rank of $q_u$. From all accepting cycles discovered this way we choose those with the maximal mean-weight.
\end{proof}

\subsection{Proof of \autoref{prop:Ltrail-coNP-comp-rbst}}\label{proof:prop:Ltrail-coNP-comp-rbst}

\begin{restatable}[]{proposition}{propLtrailcoNPcomprbst}\label{prop:Ltrail-coNP-comp-rbst}
    The problem of computing $\circlearrowright$-RC for $\valrbst_L$ is coNP-complete.
\end{restatable}

\begin{proof} 

The problem discussed in \autoref{paragraph:RC-for-the-quantitative-case-circlearrowright} is a maximum problem.
Multiplying by $-1$ the weights described  for the maximum problem we obtain a minimum problem. 
Put as a decision problem, we get:
$$L_{\trail,U}= \left\{ (G,k,v,U)\,\left|\,\begin{array}{l} G=(V,E,\theta) \mbox{ is a weighted graph, } v\in V, U\subseteq V, k \in\mathbb{Q}\\
\mbox{and \textbf{every cyclic trail }} c \mbox{ going through } v \\ \mbox{that intersects } U
\mbox{ has mean weight} \geq k
\end{array}\right.
\right\}$$

We can show that $L_{\trail,U}$ is coNP-complete by a simple reduction from $L_{\trail}$.
Given $(G,k,v)$ the reduction function returns $(G,k,v,\{v\})$.
\end{proof}

\subsection{Proof of \autoref{prop:rc-for-reg-omega-quantitative}}\label{proof:prop:rc-for-reg-omega-quantitative}

\proprcforregomegaquantitative*

\begin{proof}
    The strong $\leadsto$-RC for $\valrbst_L$ is constructed as follows. 
    For each MSCC $S$ of $\aut{A}$ we find the accepting cycle with maximal mean-weight in $S$ and let $M_S$ be the set of all edges of such a cycle in $S$. Such cycles can be found by analyzing all simple cycles of $S$. Finding all simple cycles can be done using Johnson’s algorithm \cite{johnson75} in time polynomial in the number of states, edges, and simple cycles.\footnote{Note that the number of simple cycles can be exponential in the number of the states.} To ensure the cycle is accepted we verify its minimal state-rank is even. 
    Let $\theta_S$ be the maximal mean-weight of an accepting cycle in $S$.
    We want to propagate these weights backwards to all edges and states from which $M_S$ is reachable. 
    During this backwards propagation edges and states record the maximal robustness value that can be achieved on a word read from that edge or state. The robustness value is based on the colors of the maximal mean-weight cycle reachable from this point and the colors of the highest value spoke to this cycle. 
    Once this propagation is done, the value associated with a state $v$ is the maximal robustness value of a lasso word from $v$, call it $\theta_v$. 
    With each state $v$ we associate the set of letters $\Sigma_v$ on the first letter of such a word.
    At runtime, when on state $v$, the $\leadsto$-RC recommends $\Sigma_v$.

    The strong $\circlearrowright$-RC  is constructed as follows. 
    For each edge that is contained in an MSCC we compute the maximal mean-weight of an accepting \emph{cyclic trail} that the edge participates in. Finding all trails and their mean-weight can be done in $O(2^{m \log m})$ where $m=|E|$, by examining all possible edges combinations of length bounded by $m$. To ensure the trail is accepted we check its minimal state-rank. 
    With each state $v$ we associate a set $\Sigma_v$ that consists of letters on its outgoing edges with highest mean-weight. 
    The algorithm examines all trails and not only the shortest ones since, according to the $\circlearrowright$-RC definition, it is impossible that a shorter cyclic trail that is an infix of some longer cyclic-path and that has a higher mean-weight is missed. In such case both trails include current state and once we examine the trails of this state we find the shorter one that has higher mean-weight and recommend it. 
    For states that are not part of any MSCC we set $\Sigma_v$ to the empty set.
    At runtime, when on state $v$, the $\circlearrowright$-RC recommends $\Sigma_v$.
\end{proof}

\end{document}